\documentclass[10pt,journal,compsoc]{IEEEtran}
\usepackage[draft]{todonotes}   

\newcommand{\mc}[1]{\mathcal{#1}}
\newcommand{\bs}[1]{\boldsymbol{#1}}

\ifodd 0
\newcommand{\rev}[1]{{\color{blue}#1}}
\else
\newcommand{\rev}[1]{#1}
\fi

\ifodd 0
\newcommand{\revj}[1]{{\color{red}#1}}
\else
\newcommand{\revj}[1]{#1}
\fi

\ifodd 0
\newcommand{\revjj}[1]{{\color{blue}#1}}
\else
\newcommand{\revjj}[1]{#1}
\fi

\makeatletter
\def\ps@headings{%
\def\@oddhead{\mbox{}\scriptsize\rightmark \hfil \thepage}%
\def\@evenhead{\scriptsize\thepage \hfil \leftmark\mbox{}}%
\def\@oddfoot{}%
\def\@evenfoot{}}
\makeatother
\ifCLASSINFOpdf
\else
\fi
\usepackage{amsmath, amssymb, amsthm, mathrsfs,bm}
\usepackage{graphicx}
\usepackage{multirow}
\usepackage{tabularx}
\usepackage{float}
\usepackage{tikz,pgf}
\usepackage[linesnumbered,ruled,vlined]{algorithm2e}
\usepackage{algorithmic}
\usepackage{cite}
\usepackage{url}
\usepackage{booktabs}
\usepackage{stfloats} 

\newtheorem{thm}{Theorem} 
\newtheorem{lem}{Lemma} 

\newtheorem{pps}{Proposition}
\newtheorem{defn}{Definition}

\newcommand*{\QEDB}{\hfill\ensuremath{\square}}%

\begin{document}
%
\title{Economics of Mobile Data Trading Market\vspace{-2mm}}



\author{Junlin Yu, 
        Man Hon Cheung,
        and~Jianwei Huang,~\IEEEmembership{Fellow,~IEEE}
\IEEEcompsocitemizethanks{\IEEEcompsocthanksitem J. Yu is with Ant Group, Hangzhou, China; E-mail: julian.yjl@antgroup.com. M. H. Cheung is with the Department of Computer Science, City University of Hong Kong, Hong Kong, China; E-mail: mhcheung@cityu.edu.hk. 
J. Huang (corresponding author) is with the School of Science and Engineering, The Chinese University of Hong Kong, Shenzhen, and Shenzhen Institute of Artificial Intelligence and Robotics for Society; E-mail: jianweihuang@cuhk.edu.cn.
This work is supported by the Shenzhen Institute of Artificial Intelligence and Robotics for Society and the Presidential Fund from the Chinese University of Hong Kong, Shenzhen.
\IEEEcompsocthanksitem Part of this paper was presented in \cite{r:wiopt}.
}
}

\IEEEtitleabstractindextext{
\begin{abstract}
To exploit users' heterogeneous data demands, several mobile network operators worldwide have launched the mobile data trading markets, where users can trade mobile data quota with each other. In this paper, we aim to understand the importance of data trading market (DTM) by studying the users' operator selection and trading decisions, and analyzing the operator's profit maximizing strategy. We model the interactions between the mobile operator and the users as a three-stage Stackelberg game. In Stage I, the operator chooses the operation fee imposed on sellers to maximize its profit. In Stage II, each user chooses his operator. In Stage III, each DTM user chooses his trading decisions. We derive the closed-form expression of the unique Nash equilibrium (NE) in Stages II and III, where every user proposes the same price such that the total demand matches with the total supply. 
  We further show that the Stage I's problem is convex and compute the optimal operation fee. Our analysis and numerical results show that an operator with a small initial market share can increase its profit by proposing a DTM, which is in line with the real-world situation in Hong Kong.
\end{abstract}
}
\maketitle
\thispagestyle{empty} 

%
\IEEEpeerreviewmaketitle


\section{Introduction}\label{sec:1}
\subsection{Background and Motivation}\label{sec:1a}
Due to the significant increase of video traffic on smartphones and tablets, global mobile data traffic has been growing tremendously in the past few years \cite{r:eric}. To alleviate the tension between the mobile data demand and network capacity, mobile network operators have been experimenting with several innovative pricing schemes, such as time and location dependent pricing, shared data plans, and sponsored data pricing \cite{r:ha,r:yu2,r:andrews2,r:ma}. However, these pricing schemes do not fully take
advantage of the \emph{heterogeneous demands} across all mobile users, as a user's unused portion of his month data quota will be wasted even if another user is in need of additional data in all these schemes. 

Seizing this opportunity, China Mobile Hong Kong (CMHK), a Hong Kong mobile operator with the smallest market share \cite{r:cmhkreport}, launched the 2nd exChange Market (2CM) \cite{r:cmhk} in 2014. It is a mobile data trading platform that allows its users to trade their 4G mobile data quota with each other. In this platform, a seller can list his desirable selling price on the platform, together with the amount of data to be sold (up to his monthly data quota). If there is a buyer who is willing to buy the data at the listed price, the platform will clear the transaction and transfer the corresponding amount of the data to the buyer's monthly quota limit. 

With such a \emph{data trading market} (DTM), users can make better use of their quota in their data plan. However, operators face \emph{tradeoff} in deploying the DTM. On one hand, DTM helps operators attract/retain customers by enabling them a better control and ownership of their packages \cite{r:conext}. In addition, CMHK can benefit by charging the sellers an operation fee for each GB of sold data in the form of ``transaction tax'' \cite{r:cmhk}. On the other hand, CMHK's profit may be negatively affected, because its revenue from users' overage usage is lower with the DTM. Hence, it is not clear to the operator whether it is profitable to deploy this market.

Apart from studying the profitability of the market, we are also interested in the market mechanism.
The current 2CM market mechanism is not efficient, as only a seller can list his trading price and quantity. This means that a buyer needs to frequently check the platform to see whether he is willing to buy according to the current (lowest) selling price. This motivates us to consider a multi-unit double auction mechanism based on the one adopted in stock markets \cite{r:garbade}, \cite{r:stoll}. With such a mechanism, in every time slot, a user can choose his role (seller or buyer) and submit his (selling or buying) price and quantity to the platform.\footnote{\rev{If a user is not willing to trade in a time slot, he can set his trading quantity as zero.}} The platform clears the market at a market clearing price, which the buying price of some buyers is larger than the selling price of some sellers. The sellers with very high selling prices and the buyers with very low buying prices may not get all their proposed quantities transacted. 

Under such a mechanism, the mobile network operator can obtain revenue from the DTM in two ways. First, the operator profits from the gap between the subscription fee and the service cost for each user. Second, the operator charges the sellers an \emph{operation fee} for each unit of sold data, which is unique to a DTM. Hence, it needs to decide the optimal operation fee to maximize its profit.

In this paper, we would like to understand three important questions in such a DTM:

\begin{itemize}
	\item Why does the operator propose a DTM?
	
	\item How should the operator set the operation fee to maximize its profit?
		
	\item Given a fixed operation fee, what are the equilibrium trading behaviors among the users?
			
\end{itemize}


\subsection{Contributions}\label{sec:1b}

To answer the above questions in a coherent framework, we model the interactions between the mobile operator and the users as a three-stage \emph{Stackelberg game} with different time scales, as shown in Fig. \ref{fig:1}. We assume that there are multiple mobile operators, while only one operator proposes the DTM.\footnote{To the best of our knowledge, CMHK is the only operator in Hong Kong that deploys the DTM \cite{r:cmhk}.} We refer to this operator as the \emph{DTM operator}, and its users as the \emph{DTM users}. The operator and users make decisions in different time scales (which will be discussed in details in Section \ref{sec:2}). 
  At the beginning of every subscription horizon (Stage I), the DTM operator optimizes its operation fee imposed on the sellers to maximize its profit. 
	After the DTM operator proposes the operation fee (Stage II), each user chooses his operator given the operation fee. 
	At the beginning of every data trading horizon (Stage III), the DTM users decide their roles as sellers or buyers and the corresponding trading prices and quantities given the users choices of operators. 

  Analyzing such a two-sided market is very challenging, as it involves the complicated interactions among the operators and users in three stages. Moreover, it is difficult to guarantee the existence of the Nash equilibrium (NE) for the users' decisions due to the discontinuity of their utility functions (to be discussed in details in Section \ref{sec:2b}) \cite{r:basar,r:poorm}.
	Nevertheless, we are able to characterize the optimal operation fee in Stage I and the unique NE in Stages II and III in closed-form.

\begin{figure}[t]
\centering
\includegraphics[width=9.5cm, trim = 2cm 1.5cm 0cm 1.5cm, clip = true]{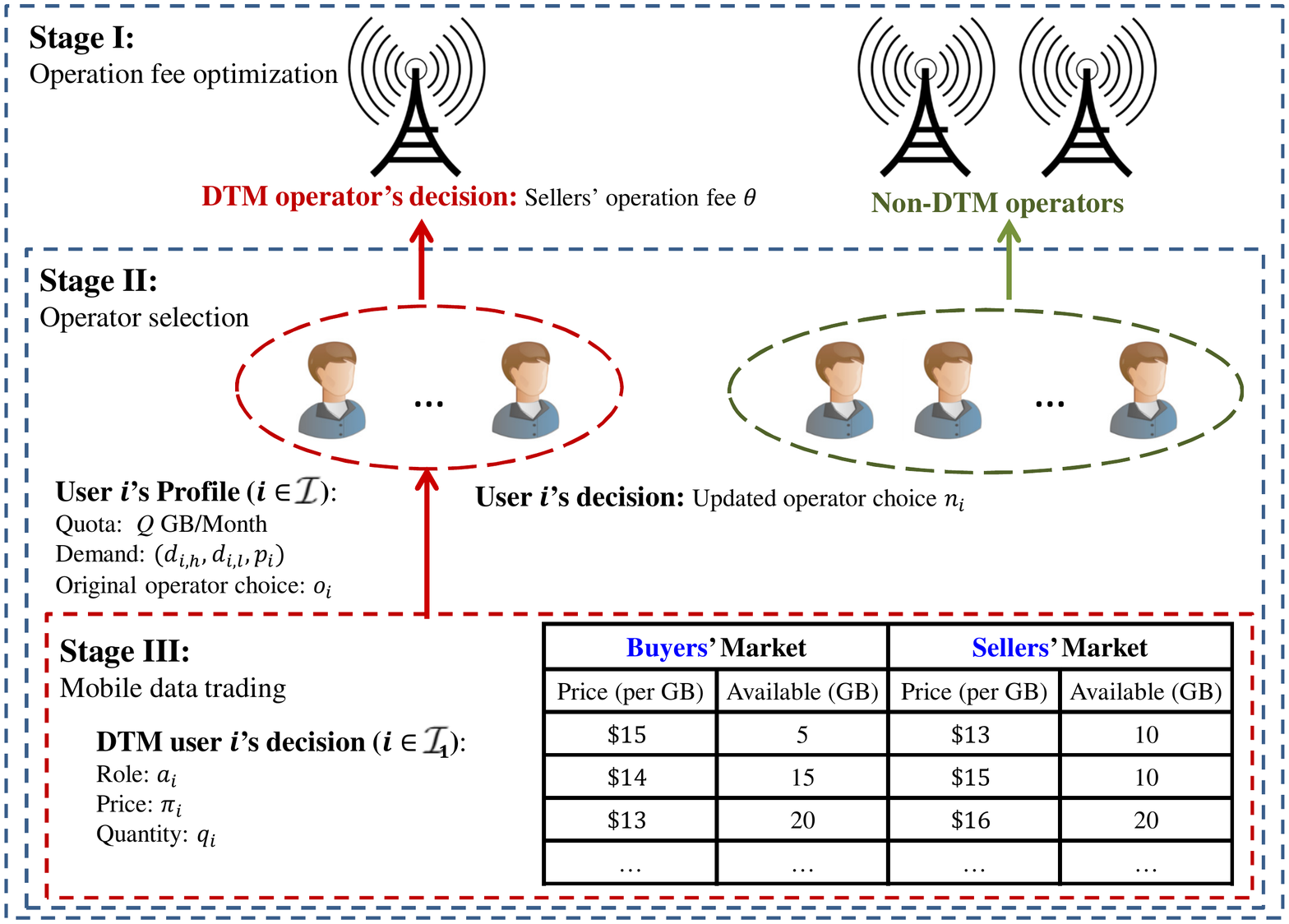} 
\vspace{-3mm}
	  \caption{Stackelberg Model of the Operator and Users' Interactions.}\label{fig:1}
	  \vspace{-5mm}
\end{figure}

In summary, our key results and contributions are as follows.
\begin{itemize}
\item \emph{Novel two-sided data trading market formulation}: Our model captures several key practical issues, such as the users' decision of trading prices and quantities without knowing in advance how much data they can sell or buy at the proposed prices. These issues have not been considered before in the context of mobile data trading. 

\item \emph{Closed-form solution of the three-stage problem}: Despite the discontinuity in the user's utility functions, we characterize the unique NE for the users' operator selection and mobile data trading game in closed-form, where different types of buyers and sellers propose the \emph{same price} such that the total supply matches the total demand in the DTM. 
  In addition, we show the operator's operation fee optimization problem is convex and derive the closed-form solution.

\item \emph{Benefit of data trading market to a small operator}: We show that an operator with a small  market share benefits from launching a DTM. This is in  line  with  the  real-world observation in Hong Kong that CMHK, which is the smallest mobile operator \cite{r:cmhkreport}, is the only DTM operator.
\end{itemize}

\subsection{Related Literature on Mobile Data Trading}\label{sec:1c}
The research on mobile DTM only emerged recently \cite{r:wanghx,r:zheng,r:zheng2,r:yu,r:andrews,wang_op18,wang_ev19}. In \cite{r:wanghx}, Wang \emph{et al.} proposed a user-initiated network for cellular users to trade data plans by leveraging personal hotspots with users' smartphones without considering the operator's participation and cooperation. In \cite{r:zheng} and \cite{r:zheng2}, Zheng \emph{et al.} studied the users' optimal bids in the market and proposed an algorithm for mobile operator to match the buyers and sellers. In \cite{r:yu}, Yu \emph{et al.} studied a single user's optimal mobile data trading problem under the future demand uncertainty from a behavioral economics perspective. However, the authors in \cite{r:zheng}, \cite{r:zheng2}, and \cite{r:yu} assumed that the sellers and buyers in the DTM can always bid prices that ensure that all their demand are satisfied and all their supply are cleared.
In \cite{r:andrews}, Andrews presented a dynamic programming problem to characterize the trading behavior of mobile users without considering the interactions between the operator and the users.
  Wang \emph{et al.} in \cite{wang_op18} studied the data plan sharing through personal hotspots and proposed a pricing mechanism that takes into account the uncertainty of mobility and sharing cost.
  Wang \emph{et al.} in \cite{wang_ev19} analyzed the economic viability of offering a DTM with rollover mechanism to explore the connection between user-flexibility and time-flexibility.

 To the best of our knowledge, our work in \cite{r:wiopt} is the first paper that studied the DTM involving the active decisions of both the operator and the users. Specifically, we formulated a two-stage decision problem, where the mobile operator determines the operation fee in the first stage, while the subscribers determine their trading decisions in the second stage.
  In this paper, as a practical extension to \cite{r:wiopt}, we further consider the \emph{market competition} between operators with and without DTM. As a result, we include a new stage regarding the users' operator selections and formulate a three-stage decision problem, which significantly complicates the problem structure. It also complicates the operation fee optimization in the first stage, as we need to keep track of the fraction of users choosing between the DTM and non-DTM operators. Nevertheless, we are able to characterize the optimal solution and the NE in closed-form.

The rest of the paper is organized as follows. In Section \ref{sec:2}, we introduce the system model. In Section \ref{sec:2b} and Section \ref{sec:3e}, we analyze the users' mobile data trading game and operator selection game in Stages III and II, respectively. In Section \ref{sec:4}, we analyze the mobile operator's operation fee optimization problem in Stage I. We present the numerical results in Section \ref{sec:sim}, and conclude in Section \ref{sec:6}.

\section{System Model}\label{sec:2}
In this work, we study the operator's operation fee optimization problem and the users' joint operator selection and mobile data trading problem comprehensively.
We first introduce the three-stage sequential interactions model in Section \ref{sec:new2a}.
Then, we introduce the users' mobile data trading platform in Section \ref{sec:new2b}.
Finally, we outline the solution of the three-stage Stackelberg game model in Section \ref{sec:new2c}.
The list of key notations are shown in Table \ref{table:variable}.

\begin{table}[t]
\renewcommand{\arraystretch}{1.2}
\caption{Key Notation}
\label{table:variable}
\begin{tabular}{ll}
\toprule
\textbf{Notation} & \textbf{Meaning} \\
\midrule
 $i$, $\mathcal{I}$, $\mathcal{I}_1$ & User index, set of users, and set of DTM users\\
\hline
$\theta$ & DTM operation fee on sellers in Stage I\\ %
\hline
$o_i$, $n_i$ & User $i$'s original operator choice  \\ 
 & and updated operator choice in Stage II \\
\hline
 $\bs{x}_i = (a_i,\pi_i,q_i)$ & DTM user $i$'s trading decision in Stage III: \\ 
 & Role $a_i$, trading price $\pi_i$, trading quantity $q_i$ \\
\hline
$\bs{r}_i(\bs{x})$ & DTM user $i$'s transaction quantity allocated \\ 
 & by the DTM mechanism (i.e., Algorithm \ref{algo:dtm}) \\
\hline
$U_{i,1}(r_i(\boldsymbol{x}))$ & DTM user $i$'s expected payoff defined in \eqref{eq:1}\\ 
\hline
$U_{i,0}$ & Non-DTM user $i$'s expected payoff defined \\ 
 & in \eqref{eq:stage2} \\
\hline
 $U_i(\bs{n})$ & User $i$'s expected payoff given all users' \\ 
 & operator choices $\bs{n}$ defined in \eqref{eq:u} \\
\hline
$\revj{Q_i}$ & \revj{User $i$'s monthly quota} \\ 
\hline
$\revj{Q}$ & \revj{Mean value of $Q_i$'s distribution} \\ 
\hline
$d_i\in\{d_{i,h}, d_{i,l}\}$ & User $i$'s demand (high or low demand) \\ 
\hline
$p_i$ & Probability that user $i$ has a high demand \\ 
\hline
$D_h$, $D_l$ & Mean values of $d_{i,h}$ and $d_{i,l}$'s distributions \\ 
\hline
 $L(\cdot)$ & Satisfaction loss function defined in \eqref{equ:satloss}\\ 
\hline
 $\kappa$ & Operator's overage usage fee imposed on users\\ 
\hline
$\hat{\pi}(\boldsymbol{n},\theta)$ & Market clearing price (Stage III's analysis) \\ 
\hline
$\boldsymbol{n}^*(\theta)$ & Users' operator choices at NE \\
 & (Stage II's analysis) \\
\hline
$\hat{\pi}(\theta) $ & Market clearing price (Combining both Stage II\\ %
 ($\triangleq \hat{\pi}(\boldsymbol{n}^*(\theta),\theta)$) & and Stage III's analysis)\\
\hline
 $\alpha$ & DTM operator's original market share\\ 
\hline
$e_i(n_i)$ & User $i$'s switching cost for choosing operator $n_i$\\ 
\hline
$P_L(\pi)$, $P_H(\pi)$ & DTM user's type thresholds at price $\pi$\\ 
\hline
$P_L'(\pi)$, $P_H'(\pi)$ & Non-DTM user's type thresholds at price $\pi$\\ 
\hline
$P(\theta)$ & DTM operator's profit given operation fee $\theta$ \\ 
 & defined in \eqref{eq:8} \\ 
\bottomrule
\end{tabular}
\end{table}

\subsection{Three-Stage Sequential Interactions Model}\label{sec:new2a}
As illustrated in Fig. \ref{fig:1}, we consider a market consisting of multiple mobile operators and a set $\mathcal{I} = \{1, ... , I\}$ of users, where one of the operators (i.e., the DTM operator) deploys a mobile DTM \cite{r:cmhk}. Among these $I$ users, a set $\mathcal{I}_1$ of users are DTM users. We use $n_i = 1$ to denote that user $i$ is a DTM user and $n_i = 0$ otherwise, so $\mathcal{I}_1 \triangleq \{j \in \mathcal{I}: n_j = 1 \}$.
  We assume that the number of users in the market is very large (e.g., $I\rightarrow\infty$), hence the impact of a single user's action on the whole population can be ignored \cite{r:asu}.

\begin{figure}[t]
\centering
\includegraphics[width=9cm, trim = 1.5cm 10cm 2.5cm 1cm, clip = true]{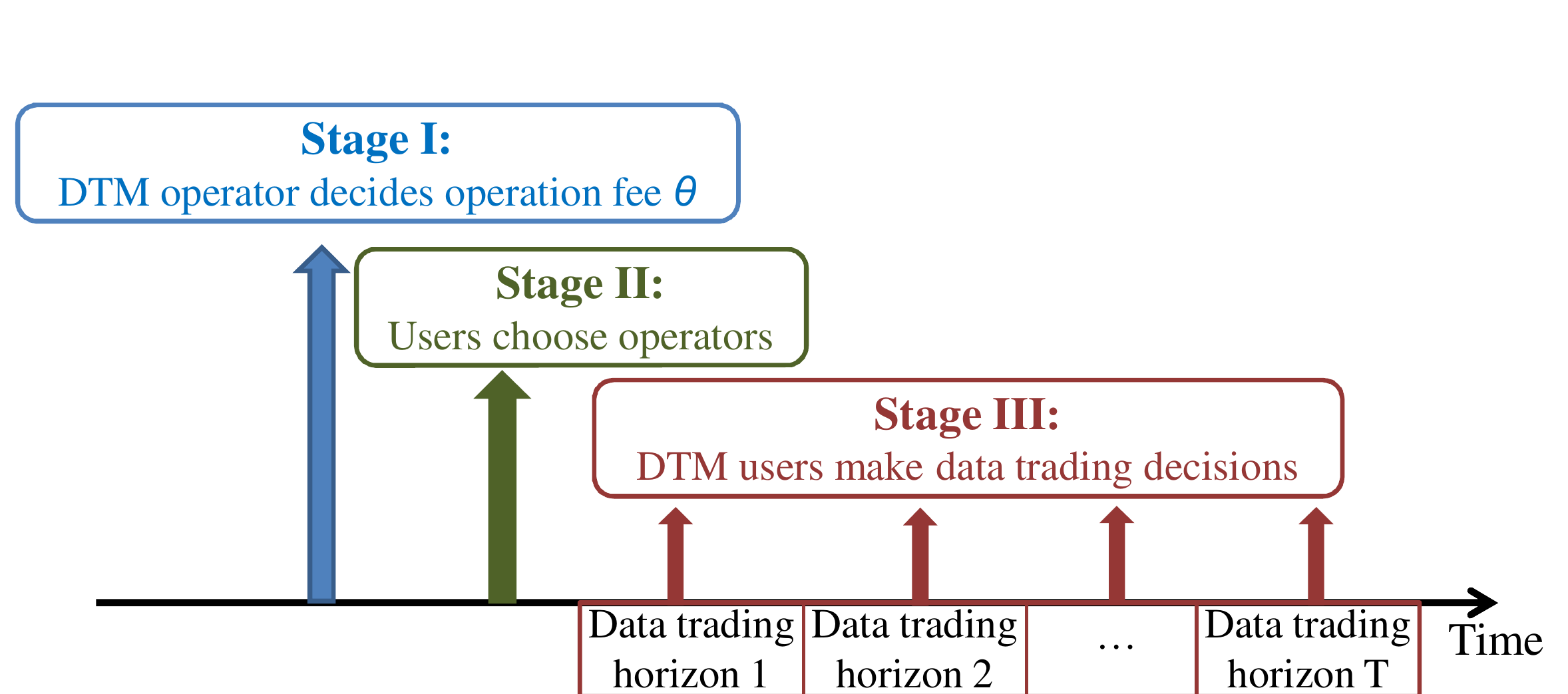}
\vspace{-3mm}
	  \caption{Timescales of the sequential decisions by the operator and the users.}\label{fig:t}
	  \vspace{-5mm}
\end{figure}

  We model the interactions between the mobile operator and the users in three stages with different time horizons as illustrated in Fig.~\ref{fig:t}. 
	First, we assume that the users' \emph{subscription horizon} is one year because the unit cost of providing services (e.g., the annual license fee for mobile spectrum \cite{r:ofcom}) usually varies annually. 
	Second, according to \cite{r:length}, the period of CMHK's data plan contract is 12 months, hence users can only change operator at the beginning of every subscription horizon. 
	Third, we assume that a data trading horizon is a billing cycle (or a month).\footnote{We have conducted a survey of 51 CMHK users regarding the trading frequency. The survey shows that about 60\% of the users (31/51) only trade once during a billing cycle. We assume that all the users trade at the beginning of the billing cycle.} Hence, one subscription horizon can be divided into $T=12$ \emph{data trading horizons}. 
The three stages are as follows:
\begin{itemize}
\item Stage I: \rev{At the beginning of every subscription horizon}, the DTM operator decides the operation fee $\theta$ charged on the sellers for each unit of sold data. 
\item Stage II: After the DTM operator has proposed the operation fee $\theta$, the users choose their operators $\boldsymbol{n}=(n_1,\ldots,n_I)$ for the entire subscription horizon. 
\item Stage III: At the beginning of each data trading horizon, the DTM users decide their trading decisions $\boldsymbol{x}=(\boldsymbol{x}_1,\ldots,\boldsymbol{x}_I)$.\footnote{\rev{If we allow each user to trade multiple times in a billing cycle, we can formulate the DTM user’s decision process as a dynamic programming (DP) problem \cite{wang_ev19}.} \revj{However, it is challenging to derive the DTM users' equilibrium trading decisions (in Stage III) in closed-form, so it would significantly complicate the analysis in Stages I and II. Thus, we leave the general model with multiple trading periods as a future work.}} The users' decisions on choices of operators $\boldsymbol{n}$ are fixed within each subscription horizon. 
\end{itemize}

\subsection{Mobile Data Trading Platform}\label{sec:new2b}
\subsubsection{Users' Decisions}
For DTM user $i \in \mathcal{I}_1$, his \emph{trading decision} is defined as $\boldsymbol{x}_i=(a_i, \pi_i, q_i)$, which  consists of three components. First, user $i$ needs to decide his \emph{role} $a_i\in\{s,b\}$, i.e., whether to be a seller ($a_i=s$) or a buyer ($a_i=b$), or not participate in the market.
 Correspondingly, he has to determine his \emph{trading price} $\pi_{i}$ and his \emph{trading quantity} $q_{i}$. Specifically, if user $i$ chooses to be a seller, the price $\pi_i$ and the quantity $q_i$ refer to his selling price and selling quantity, respectively. On the other hand, if user $i$ chooses to be a buyer, the price $\pi_i$ and the quantity $q_i$ refer to his buying price and buying quantity, respectively.  \revjj{In the DTM, the proposed prices must be smaller than the usage-based price $\kappa$, otherwise no transaction will happen in the market.} Hence, we have $\pi_i \in \Pi = [0, \kappa]$ and $q_i \in [0, \infty)$. 
  Let $\bs{x}_i = (a_i,\pi_i,q_i)$ be the trading decision of DTM user $i \in \mc{I}_1$ and $\bs{x} = (\bs{x}_i, i \in \mc{I}_1)$ be all the DTM users' trading decisions.
  We define the set of feasible strategies of DTM user $i$'s trading decisions as 
\begin{equation}
  \mathcal{X}_i =\{(a_i,\pi_{i},q_i): a_i\in\{s,b\},\pi_i\in[0,\kappa],q_i\in[0,\infty)\}.
\end{equation}
%


\subsubsection{Sellers' and Buyers' Markets}\label{sec:2b1}
We consider a two-sided mobile data trading platform for the DTM users in set $\mathcal{I}_1$. It is based on the first-price multi-unit double auction mechanism \cite{r:stoll}, which consists of four main steps\footnote{For the details of the mechanism, please refer to Appendix \ref{app:a}.}:

\begin{itemize}
\item Step 1 (Bidding): At the beginning of every data trading horizon, all users make their trading decisions and submit their \emph{bids}, which include their roles $ \boldsymbol{a} = (a_i, \forall i \in \mathcal{I}_1)$, trading prices $ \boldsymbol{\pi} = (\pi_i, \forall i \in \mathcal{I}_1)$, and trading quantities $ \boldsymbol{q} = (q_i, \forall i \in \mathcal{I}_1)$, simultaneously to the platform. 

\item Step 2 (Prioritization): The platform sorts the bids of the users in the sellers' and buyers' markets, respectively, according to their proposed prices.
For the example in Fig.~\ref{fig:1}, in the table of mobile DTM in Stage III, the total demand in the buyers' market at the highest buying price of \$15 is 5 GB, which will be satisfied with the highest priority. The total supply in the sellers' market with the lowest selling price of \$13 is 10 GB, which will be sold with the highest priority.

\item Step 3 (Allocation): The platform clears the markets by allocating the bids through an auction mechanism with the priority orders, and outputs all users' transacted quantities $\boldsymbol{r}(\boldsymbol{x})=(r_i(\boldsymbol{x}),\forall i \in \mathcal{I}_1)$. Notice that the selling bids are only allocated to the buying bids whose buying price is no smaller than the corresponding selling price. For example, the platform first allocates the top row's bids (5 GB data quota supply with the selling price of \$13 to 5 GB data quota demand with the buying price of \$15), and then allocates the remaining 5 GB supply with the selling price of \$13 on the first row to 5 GB demand with the buying price of \$14 on the second row.  
The remaining 10 GB demand with the buying price of \$14, however, is not satisfied, because the remaining supply on the second row is with the selling price of \$15. If the 15 GB of demand with the buying price of \$14 is from multiple buyers, the buyers will \emph{equally share} the supplies. The details of the equally share allocation mechanism among multiple buyers are discussed in Algorithm \ref{algo:dtm} in Appendix \ref{app:a}.

\item Step 4 (Payment): Finally, the platform decides the payment of the users. If there is a gap between the selling and buying prices of the transacted data, then the price gap leads to the revenue of the operator. For example, the operator gains \$2 for allocating each GB with the selling price of \$13 to that with the buying price of \$15. The platform will also charge a seller $\theta$ dollars of operation fee for each GB sold.

\end{itemize}

At the end of every data trading horizon, all users can access the aggregate information on other users' decisions of the last data trading horizon, by checking the updated market information on the platform. We assume that by learning from the historical market information, all the users know the distribution of user types in the market.


\subsection{Outline of Analysis}\label{sec:new2c} 
  In the following three sections, we will use the \emph{backward induction} to solve the \rev{three-stage sequential game}.\footnote{\rev{Backward induction is the standard technique to analyze a sequential game, where we analyze the problem in the order of Stages III, II, and I. As a result, we discuss the analysis of the three stages in reverse order from Section \ref{sec:2b} to Section \ref{sec:4}.}} We start from the DTM users' mobile data trading game of Stages III in Section \ref{sec:2b}, and then study the users' operator selection game of Stage II in Section \ref{sec:3e}. Finally, we discuss the operator's operation fee optimization problem of Stage I in Section \ref{sec:4}. 

\section{Stage III: Mobile Data Trading Game}\label{sec:2b}
In this section, we study the users' mobile data trading game in Stage III, given the operator's operation fee $\theta$ and the users' operator selection $\boldsymbol{n}=(n_1,\ldots,n_I)$. We first introduce the DTM user model in Section \ref{sec:new32}. Next, we formulate the DTM users' mobile data trading game in Section \ref{sec:3a}, and analyze the DTM users' game equilibrium in Section \ref{sec:3d}. 

\subsection{\rev{Model: DTM User}}\label{sec:new32}
\subsubsection{Quota and Demand}\label{sec:2a1}

\revj{Let $Q_i$ be the \emph{monthly quota} of user $i$. For the ease of exposition, we assume that there are two possible realizations of a user $i$'s data demand for the month (or simply called \emph{demand}): $d_i\in\{d_{i,h}, d_{i,l}\}$, with $0<d_{i,l}<Q_i<d_{i,h}$.\footnote{The analysis for the case where both $d_{i,h}$ and $d_{i,l}$ are higher (or lower) than the monthly quota \revj{$Q_i$} is relatively trivial, and hence is omitted here due to space limitations.}\footnote{\rev{It is possible that there will be no trade when all the DTM users are sellers (i.e., when $d_{i,h} < Q_i, \forall i \in \mc{I}_1$) or buyers (i.e., when $d_{i,l} > Q_i, \forall i \in \mc{I}_1$). However, these extreme scenarios are not realistic in practice, as the operator has the freedom to adjust the quota $Q_i$ such that they will not happen. It is also possible that no trade happens when the DTM sets a very high operation fee $\theta$. However, in reality, the DTM operator can reduce operation fee to facilitate trading.}} The probability for user $i$ to observe a high demand $d_{i,h}$ is $p_i$, and the probability of observing a low demand $d_{i,l}$ is $1-p_i$.\footnote{\rev{In other words, each DTM user only knows his demand distribution, but not the actual demand, when making his trading decision.}} 
We assume that the distributions of $p_i$, $Q_i$, $d_{i,l}$, and $d_{i,h}$ are mutually independent. In addition, we assume that $Q_i$, $d_{i,l}$ and $d_{i,h}$ can follow any distributions with the mean values of $Q$, $D_l$ and $D_h$, respectively. 
}

\subsubsection{Satisfaction Loss}\label{sec:2a2}
Each user $i$ will incur a \emph{satisfaction loss} when his demand $d_i\in\{d_{i,h},d_{i,l}\}$ exceeds his monthly data quota $\revj{Q_i}$. We consider a linear satisfaction loss function
\begin{equation} \label{equ:satloss}
     L(\revj{Q_i}-d_i)=-\kappa[d_i-\revj{Q_i}]^+,
\end{equation}
where $[z]^+ = \max \{0,z\}$. Here, $[d_i - \revj{Q_i}]^+$ is the amount of insufficient data. When $d_i-\revj{Q_i}$ is positive, it means that the quota is exceeded. The linear coefficient $\kappa$ represents the overage usage fee imposed by the mobile operator in a two-part pricing tariff, where the user pays a fixed fee for the data consumption up to a monthly quota and a linear usage-based cost for any overage data consumption. Such a pricing model is widely used by major mobile operators. For example, for a 4G CMHK user, $\kappa = \$60$ with a monthly data quota of 1 GB. By selling or buying data in the market, a user can change his \emph{effective remaining data quota} (for the current month only), and hence will change his expected satisfaction loss.

\subsection{\rev{Problem Formulation: DTM Users' Non-cooperative Game}}\label{sec:3a}
First, we define a DTM user $i$'s payoff in Stage II in (\ref{eq:1}), which equals to his expected satisfaction loss minus the net payment of the trade and the switching cost.\footnote{Since we assume the data plan by all operators are with same quota and same subscription fee, it will not affect the users' decisions. Hence, we omit the subscription fee in the user's payoff function in (\ref{eq:stage2}) and (\ref{eq:1}).} 
For the seller case ($a_i=s$), the first term $(\pi_{i}-\theta)r_i((s,\pi_i,q_i),\boldsymbol{x_{-i}})$ is the revenue from selling data. The second and third terms correspond to the seller's expected satisfaction loss after selling $q_{i}$ data quota under the high and the low demand realization, respectively. Notice that each seller has to pay the operation fee of $\theta$ to the mobile operator for each unit of transacted data.\footnote{\rev{To benefit from reselling, a DTM user needs to first buy the data at a very low price and later resells it at a very high price, such that the difference between the selling price and the buying price is no less than the DTM’s operation fee $\theta$ (on sellers). This condition is not easy to satisfy in practice, so the reselling behavior in DTM is not common.}} 
In addition, if user $i$ is switching to another operator, he is going to pay an additional switching cost $e_i(n_i)$, where we will define in more details in (\ref{eq:epsilon}) in Section \ref{sec:2b} related to choice of operator. 
The buyer's payoff function on the last line of (\ref{eq:1}) is similar, except that the operator does not charge the buyer an operation fee.
\revjj{Notice that the payoff function in (\ref{eq:1}) is discontinuous. For example, if a seller's supply is only partially cleared, he can clear all his supply by decreasing his selling price by a unit $\epsilon$\footnote{We assume that $\epsilon$ is the smallest price unit used in the mobile data trading platform. For example, $\epsilon=1$ HKD in 2CM. }, and hence makes a discontinuous increase in his payoff.}

\begin{figure*}[ht]
\begin{equation}\label{eq:1}
\begin{split}
&U_{i,1}(r_i(\boldsymbol{x_i},\boldsymbol{x_{-i}}))=U_{i,1}\left(r_i\left((a_i,\pi_i,q_i),\boldsymbol{x_{-i}}\right)\right)\!\\
&=\!\left\{\!
    \begin{aligned}
    & (\pi_{i}\!-\!\theta) r_i((s,\pi_i,q_i),\boldsymbol{x_{-i}})\!+\!p_i L(\revj{Q_i}\!-\!r_i((s,\pi_i,q_i),\boldsymbol{x_{-i}})\!-\!d_{i,h})\!+\!(1\!-\!p_i)\!L(\revj{Q_i}\!-\!r_i((s,\pi_i,q_i),\boldsymbol{x_{-i}})\!-\!d_{i,l})\!-\!e_i(1),  \text{ if $a_i=s$},\\
	&-\pi_{i}r_i((b,\pi_i,q_i),\boldsymbol{x_{-i}})+\!p_i L(\revj{Q_i}\!+\!r_i((b,\pi_i,q_i),\boldsymbol{x_{-i}}))\!-\!d_{i,h})\!+\!(1\!-\!p_i)\!L(\revj{Q_i}\!+\!r_i((b,\pi_i,q_i),\boldsymbol{x_{-i}}))\!-\!d_{i,l})\!-\!e_i(1),	\text{ if $a_i=b$}.
    \end{aligned}
        \right.
\end{split}
\end{equation}
\hrulefill
\vspace{-1mm}
\end{figure*}

Next, we model the DTM users' interactions as the following non-cooperative game:
\begin{defn}
  A mobile data trading game is a tuple $\Omega = (\mathcal{I}_1, \mathcal{X}, \boldsymbol{U}_1)$ defined by
\begin{itemize}
    \item \emph{Players}: The set $\mathcal{I}_1$ of DTM users. 
    \item \emph{Strategies}: \revjj{Each player chooses an action (pure strategy) $\boldsymbol{x_i}=(a_i, \pi_i, q_i)\in\mathcal{X}_i$, which is his bid to the platform. }
    The strategy profile of all the players is $\boldsymbol{x} = (\boldsymbol{x}_i,\forall i \in\mathcal{I}_1)$ and the set of feasible strategy profile of all the players is $\mathcal{X} = \mathcal{X}_1 \times \ldots \times \mathcal{X}_{I_1}$.
    \item \emph{Payoffs}: The vector $\boldsymbol{U_1} = (U_{i,1}, \forall i \in \mathcal{I}_1)$ contains all users' payoffs as defined in (\ref{eq:1}).\footnote{\rev{We assume a non-cooperative game with complete information, where the users may derive each other's type distribution through learning.}}
\end{itemize}
\end{defn}

\subsection{\rev{Analysis: Nash Equilibrium}}\label{sec:3d}
We first investigate a user's best response in Section \ref{sec:3BR}, and then characterize the NE in Section \ref{sec:3NE}.

\subsubsection{Best Response Analysis}\label{sec:3BR}
 We first define a user's best response as the strategy that maximizes his payoff given the fixed strategies of other users.

\begin{defn}
User $i$'s best response is
\begin{equation}
\begin{split}
\boldsymbol{x}_i^{BR}(\boldsymbol{x}_{-i})&\triangleq(a_i^{BR}(\boldsymbol{x_{-i}}),\pi_i^{BR}(\boldsymbol{x_{-i}}),q_i^{BR}(\boldsymbol{x_{-i}}))\\
&=\arg\max\limits_{\boldsymbol{x_i}\in\mathcal{X}_i}U_{i,1}(r_i(\boldsymbol{x_{i}},\boldsymbol{x_{-i}})).
\end{split}
\end{equation}

\end{defn}

 To characterize the best response, we then define the \emph{transaction selling price} $\hat{\pi}_s$ and the \emph{transaction buying price} $\hat{\pi}_b$
for a given strategy profile $\boldsymbol{x} = (\boldsymbol{x}_i, \forall i\in\mathcal{I}_1)$ in Definitions 3 and 4, respectively. Here, we use $\epsilon>0$ to denote the smallest price unit. When a user makes a decision, he does not need to know the choice of each of the other users, but only needs to know the accumulated bids (e.g., the total demands and supplies in terms of GBs at each price). 
This allows us to analyze the users' best responses based on  $\hat{\pi}_s$ and $\hat{\pi}_b$ instead of $\boldsymbol{x}$.

\begin{figure}[t]
\centering
\includegraphics[width = 0.4\textwidth]{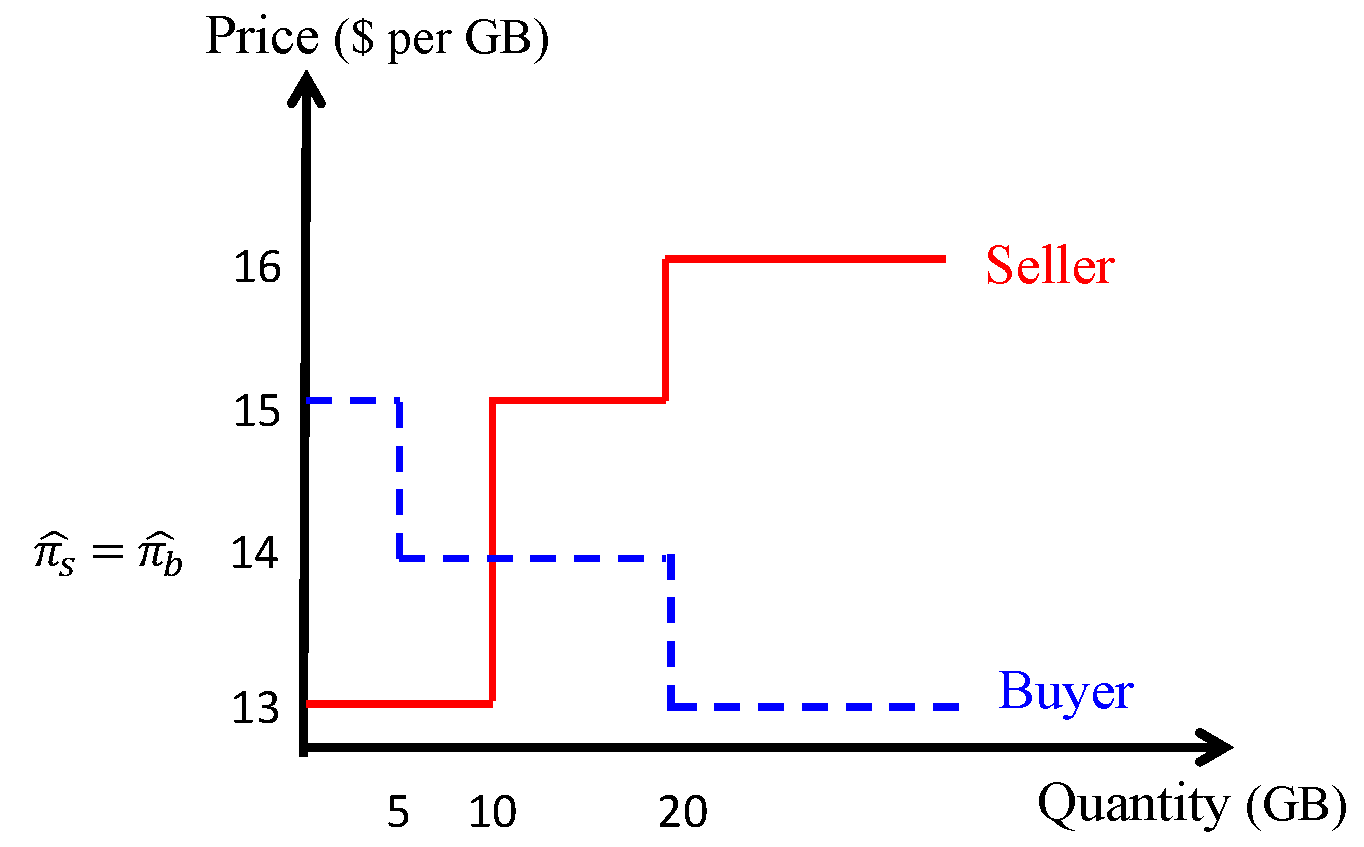}
\vspace{-3mm}
	  \caption{An example of $ \hat{\pi}_{s}$ and $ \hat{\pi}_{b}$.}\label{fig:2}
	  \vspace{-5mm}
\end{figure}

\begin{defn}
The \emph{transaction selling price}\footnote{Based on the definition, the sellers who propose this price can get \emph{some} or \emph{all} of their selling quantities transacted. The sellers who propose selling prices lower than this price will get \emph{all} their selling quantities transacted, because they have higher priorities. Notice that lower priority bids can only be cleared when all the higher priority bids are cleared.} $\hat{\pi}_{s}$ corresponds to the minimum price such that a seller cannot get any of his selling quantity transacted, if his proposed selling price is one unit larger than the transaction selling price:
\begin{equation}\label{eq:stp}
    \hat{\pi}_{s}\!=\!\min\{\pi_{i}\!:\!r_i((s,\pi_{i}+\epsilon,q_{i}),\boldsymbol{x_{-i}})=0, \forall i \in \mathcal{I}_1\}.
\end{equation}
\end{defn}
\begin{defn}
The \emph{transaction buying price} $\hat{\pi}_{b}$ corresponds to the maximum price such that a buyer cannot get any of his buying quantity transacted, if his proposed buying price is one unit smaller than the transaction buying price $\hat{\pi}_{b}$:
\begin{equation}\label{eq:btp}
    \hat{\pi}_{b}\!=\!\min\{\pi_{i}\!:\!r_i((b,\pi_{i}-\epsilon,q_{i}),\boldsymbol{x_{-i}})=0, \forall i  \in \mathcal{I}_1\}.
\end{equation}
\end{defn}

An example of the definitions is shown in Fig. \ref{fig:2}, based on the market supply and demand shown in Fig. \ref{fig:1}. The platform sorts user's bids according to their prices, then clears the bids with the highest priorities first. In this example, the transaction selling price $\hat{\pi}_s=14$ and the transaction buying price $\hat{\pi}_b=14$, because the users who propose a price of \$14 can get some quantity transacted, and the sellers who propose one unit higher (\$15) and the buyers who propose one unit lower (\$13) cannot get any quantity transacted.

In the following proposition, we characterize user $i$'s  best response given other users' strategies $\boldsymbol{x_{-i}}$. 

\begin{pps} \label{prop:br}
The best response $\boldsymbol{x}_i^{BR}(\boldsymbol{x}_{-i})$ of user $i$ with type $(p_i, d_{i,h}, d_{i,l})$ is given in (\ref{eq:thm2}).
\end{pps}

\begin{figure*}[htp]
\begin{align}\label{eq:thm2}
     &\boldsymbol{x}_i^{BR}(\boldsymbol{x}_{-i})\notag\\
     &\left\{
    \begin{aligned}
    &\!\!\!\!=\!\!(\!s,\hat{\pi}_s,\revj{Q_i}\!\!-\!d_{i,l}\!), ~~~~~\text{if}~  p_i \!\leq\! \frac{\hat{\pi}_s \!\!-\! \theta}{\kappa} ~~~~~~\text{ and }~\frac{\sum_{j\in\mathcal{HB}_i}\!q_j\!\!-\!\!\sum_{j\in\mathcal{LS}_i}\!q_j\!\!-\!\!\sum_{j\in\{k:k\in\mathcal{E}_i,q_k<\revj{Q_i}-d_{i,l}\}}\!q_j}{|\{k:k\in\mathcal{E}_i,q_k\geq \revj{Q_i}-d_{i,l}\}|}\!\geq\! \revj{Q_i}\!-\!d_{i,l},\\
    &\!\!\!\!=\!\!(\!s,\hat{\pi}_s\!-\!\epsilon,\revj{Q_i}\!\!-\!d_{i,l}\!), ~\text{if}~  p_i \!\leq\! \frac{\hat{\pi}_s \!\!-\!\epsilon\!-\! \theta}{\kappa} ~\text{ and }~\frac{\sum_{j\in\mathcal{HB}_i}\!q_j\!\!-\!\!\sum_{j\in\mathcal{LS}_i}\!q_j\!\!-\!\!\sum_{j\in\{k:k\in\mathcal{E}_i,q_k<\revj{Q_i}-d_{i,l}\}}\!q_j}{|\{k:k\in\mathcal{E}_i,q_k\geq \revj{Q_i}-d_{i,l}\}|}\!<\! \revj{Q_i}\!-\!d_{i,l},\\ 
    &\!\!\!\!=\!\!(\!b,\hat{\pi}_{b},d_{i,h}\!\!-\!\revj{Q_i}\!), ~~~~~\text{if}~ p_i\!\geq\!\frac{\hat{\pi}_{b}}{\kappa} ~~~~~~~~ \text{ and } ~~ \frac{\sum_{j\in\mathcal{LS}_i}\!q_j\!\!-\!\!\sum_{j\in\mathcal{HB}_i}\!q_j\!\!-\!\!\sum_{j\in\{k:k\in\mathcal{E}_i,q_k<d_{i,h}-\revj{Q_i}\}}\!q_j}{|\{k:k\in\mathcal{E}_i,q_k\geq d_{i,h}-\revj{Q_i}\}|}\!<\! d_{i,h}\!-\!\revj{Q_i},\\        &\!\!\!\!=\!\!(\!b,\hat{\pi}_{b}\!\!+\!\epsilon,d_{i,h}\!\!-\!\revj{Q_i}\!), ~~\text{if}~ p_i\!\geq\!\frac{\hat{\pi}_{b}\!\!+\!\epsilon}{\kappa} ~~~~~ \text{ and } ~ \frac{\sum_{j\in\mathcal{LS}_i}\!q_j\!\!-\!\!\sum_{j\in\mathcal{HB}_i}\!q_j\!\!-\!\!\sum_{j\in\{k:k\in\mathcal{E}_i,q_k<d_{i,h}-\revj{Q_i}\}}\!q_j}{|\{k:k\in\mathcal{E}_i,q_k\geq d_{i,h}-\revj{Q_i}\}|}\!\geq\! d_{i,h}\!\!-\!\revj{Q_i},\\ 
		    &\!\!\rev{=(s,0,0) \text{ or } (b,0,0)}, ~~\text{otherwise}.    
    \end{aligned}
    \right.
\end{align}\label{thm:2}
\hrulefill
\vspace{-1mm}
\end{figure*}

\vspace{-2mm}

The proof of Proposition \ref{prop:br} is given in Appendix \ref{app:b}. 

First, Proposition \ref{prop:br} states that every user who wants to trade will either choose to sell the quota at the quantity $\revj{Q_i}-d_{i,l}$ or choose to buy at the quantity $d_{i,h}-\revj{Q_i}$ if they choose to join DTM, due to the piecewise linearity of (\ref{eq:1}) in quantity $q_i$.

Next, we discuss the five lines of equation (\ref{eq:thm2}) in details:
\begin{itemize}
    \item The first and third lines: According to the transaction functions in (\ref{eq:realization1}) and (\ref{eq:realization2}), if the supply is not enough, the users with the same price will equally share the supply. In this case, the bids of users with a low $q_i$ will be fully satisfied\footnote{By saying ``the bid of a user is satisfied'', we mean that ``his demand is satisfied'' if the user is a buyer, or ``his supply is cleared'' if the user is a seller.}, while the bids of users with a high $q_i$ will only be partially satisfied. Hence, the users with a low $p_i$ and a low $\revj{Q_i}-d_{i,l}$ will choose to sell with a lower price $\hat{\pi}_s$ (as shown in the first line of (\ref{eq:thm2})), and the users with a high $p_i$ and a low $d_{i,h}-\revj{Q_i}$ will choose to buy with a higher price $\hat{\pi}_b$ (as shown in the third line of (\ref{eq:thm2})).
    \item The second and fourth lines: To make their bids fully satisfied, the users with a higher quantity can propose a price with a slightly higher priority. Hence, the users with a low $p_i$ and a high $\revj{Q_i}-d_{i,l}$ will choose to sell with a lower price $\hat{\pi}_s-\epsilon$ (as shown in the second line of (\ref{eq:thm2})), and the users with a high $p_i$ and a high $d_{i,h}-\revj{Q_i}$ will choose to buy with a higher price $\hat{\pi}_b+\epsilon$ (as shown in the fourth line of (\ref{eq:thm2})).
    \item The fifth line: \rev{A user who is not willing to participate will propose a zero price and quantity as a seller or a buyer.} 
\end{itemize}



\subsubsection{Unique NE Characterization}\label{sec:3NE}
 Next, we define the NE as the intersection of all the users' best response correspondences.
\begin{defn}(Nash Equilibrium (NE)):
A strategy profile $\boldsymbol{x}^*$ is an NE if and only if
\begin{align}\label{eq:def4}
    U_{i,1}(r_i(\boldsymbol{x}_i^*,\boldsymbol{x_{-i}^*}))&\geq U_{i,1}(r_i(\boldsymbol{x_i'},\boldsymbol{x_{-i}^*})), \forall \boldsymbol{x_i'}\in \mathcal{X}_i, i\in\mathcal{I}_1.
\end{align}
\end{defn}

If a strategy profile is a NE, none of the users has the incentive to change his strategy, and the transacted prices will not change.  
To obtain the NE, we first show the conditions that any NE should satisfy in Lemma \ref{lem:1}. 
\begin{lem}\label{lem:1}
At any NE $\boldsymbol{x}^*$, there exists a unique transaction price $\hat{\pi}(\boldsymbol{n},\theta)$ such that for any fixed $\boldsymbol{n}$ and $\theta$, the following conditions are satisfied.
\begin{align}
&\sum_{i\in\{j: a_j^*=s,\pi_j^*\leq\hat{\pi}(\boldsymbol{n},\theta)\}}q_i=\sum_{i\in\{j: a_j^*=b,\pi_j^*\geq\hat{\pi}(\boldsymbol{n},\theta)\}}q_i,\label{eq:lem11}
\\
&\pi_i^{*}=\hat{\pi}(\boldsymbol{n},\theta), \forall i\in\{j\in\mathcal{I}_1:p_j<\frac{\hat{\pi}(\boldsymbol{n},\theta)-\theta}{\kappa} \notag\\
&\quad\quad\quad\quad\quad\quad\quad\quad\quad\quad\quad\quad\quad\quad\text{ or } p_j>\frac{\hat{\pi}(\boldsymbol{n},\theta)}{\kappa}\}.\label{eq:lem12}
\end{align}
\end{lem}
Equation (\ref{eq:lem11}) implies that there exists an equilibrium price $\hat{\pi}(\boldsymbol{n},\theta)$ such that the market is \emph{cleared}, where the total supply matches the total demand.
From (\ref{eq:lem12}), we observe that those who want to make a transaction in the market will propose the \emph{same price} in their best responses.\footnote{The users with type $p_i<\frac{\hat{\pi}(\boldsymbol{n},\theta)-\theta}{\kappa}$ are sellers and the users with type $p_i>\frac{\hat{\pi}(\boldsymbol{n},\theta)}{\kappa}$ are buyers. The sellers' transaction selling price equals to the buyers' transaction buying price.} The proof of Lemma \ref{lem:1} is in Appendix \ref{app:c}.

By jointly solving equations (\ref{eq:lem11}) and (\ref{eq:lem12}) in Lemma \ref{lem:1}, we obtain the unique NE in Theorem 1.
For notation convenience, we first define 
\begin{align}
P_L(\pi) \triangleq \frac{\pi-\theta}{\kappa}, \label{eq:a}
\end{align}
and
\begin{align}
P_H(\pi) \triangleq \frac{\pi}{\kappa},
\end{align}
\noindent which are the boundary indicators of user types. With these indicators, we can divide the users into different groups according to their types.  
\begin{thm}\label{co:1}
The unique NE $\boldsymbol{x}^*$ of game $\Omega$ is shown in Table \ref{table:co1}, where the equilibrium price $\hat{\pi}(\boldsymbol{n},\theta)$ is the unique solution of the following equation:
\begin{equation}\label{eq:eqp}
\sum_{i\in\{j: p_j\leq P_L(\hat{\pi}(\boldsymbol{n},\theta))\}}(\revj{Q_i}-d_{i,l})=\sum_{i\in\{j: p_j\geq P_H(\hat{\pi}(\boldsymbol{n},\theta))\}}(d_{i,h}-\revj{Q_i}).
\end{equation}
\end{thm}
Theorem \ref{co:1} states that at the NE, all the users who participate in the market will propose the \emph{same price}, and all their proposed quantities will be transacted. 
This is because the sellers proposing higher prices cannot get their selling quantities transacted, and the sellers proposing lower prices receive lower payoffs. 
The same intuition applies to the buyers. 
The users who are not willing to participate \rev{will propose a zero price and quantity}, and will not affect the market trading results. 
The proof of Theorem \ref{co:1} is given in Appendix \ref{app:d}. 

\begin{table}[t] 
\centering 
\caption{The equilibrium users' decisions on roles, quantities, and prices}
\begin{tabular}{|@{ }c@{ }|@{ }c@{ }|@{ }c@{ }|@{ }c@{ }|}
\hline
\textbf{User Type $p_i$} & \textbf{Role} $a_i^{*}$ & \textbf{Quantity} $q_i^{*}$ & \textbf{Price} $\pi_i^{*}$ \\
\hline
$0<p_i \leq P_L(\hat{\pi}(\boldsymbol{n},\theta))$ & $s$ & $\revj{Q_i}-d_{i,l}$ & $\hat{\pi}(\boldsymbol{n},\theta)$\\
\hline
 $P_L(\hat{\pi}(\boldsymbol{n},\theta))<p_i$&  \multirow{2}{*}{\text{$s$ or $b$}}  & \multirow{2}{*}{\rev{$0$}} & \multirow{2}{*}{\rev{$0$}}\\
$<P_H(\hat{\pi}(\boldsymbol{n},\theta))$ & & & \\
\hline
$\;\; P_H(\hat{\pi}(\boldsymbol{n},\theta)) \leq p_i<1$ & $b$ & $d_{i,h}-\revj{Q_i}$ & $\hat{\pi}(\boldsymbol{n},\theta)$\\
\hline
\end{tabular}
 \label{table:co1}
 \vspace{-3mm}
\end{table}

\section{Stage II: Operator Selection Game}\label{sec:3e}
In this section, we study the users' operator selection game in Stage II given the operation fee $\theta$. We introduce the user model in Section \ref{sec:user2}, and formulate the operator selection game in Section \ref{sec:3ea}. Finally, we analyze the NE in Section \ref{sec:3ec}. 

\subsection{\rev{Model: User's Operator Choice}}\label{sec:user2}
Let $o_i$ be user $i$'s \emph{original} choice of operator in the previous subscription horizon.
If user $i$ is a subscriber of DTM operator (i.e., user $i$ is a DTM user), we denote it as $o_i=1$. Otherwise, user $i$ is a subscriber of other operators, and we denote it as $o_i=0$. Thus, DTM's original market share is 
\begin{equation}
    \alpha \triangleq \frac{|\{i:o_i=1\}|}{I}.
\end{equation}

At the beginning of a new subscription horizon, user $i$ needs to decide an \emph{updated choice of operator} $n_i$, where $n_i\in\{0,1\}$. Notice that a user $i$ will incur a \emph{switching cost} $e_i(n_i)$ for the following subscription horizon linear to his expected usage\footnote{The switching cost can be interpreted as the loss of the ``repeat purchase welfare'' \cite{r:switch}\cite{r:duan}. According to \cite{r:length}, the users who repeatedly purchase his data plan can enjoy a better service (e.g., 3G service upgraded to 4G service, or a higher Internet speed due to the technology development) with the service price unchanged. A user with a higher usage can benefit more from the ``repeat purchase welfare''. However, if the user switch to the other operator, he can no longer enjoy the ``repeat purchase welfare''.} if he choose to switch an operator (i.e., $n_i\neq o_i$). Specifically, we have
\begin{equation}\label{eq:epsilon}
    e_i(n_i)= \left\{
    \begin{aligned}
    &e[p_id_{i,h}+(1-p_i)d_{i,l}], \quad \text{ if } n_i \neq o_i,\\
    &0,\quad \quad \quad\quad \quad \quad\quad \quad \quad~\quad\text{ if } n_i = o_i, \\
    \end{aligned}
    \right.
\end{equation}
where $e \geq 0$ is the switching cost parameter.

\subsection{\rev{Problem Formulation: Users' Operator Selection Game}}\label{sec:3ea}
We first define user $i$'s payoff function if he chooses the other operators (i.e., $n_i=0$) in Stage II. In this case, since no DTM is available in the other operators, his payoff equals to his expected satisfaction loss minus the switching cost:
\begin{align}\label{eq:stage2}
U_{i,0}=p_i L(\revj{Q_i}\!-\!d_{i,h})+\!(1\!-\!p_i)L(\revj{Q_i}\!-\!d_{i,l})-e_i(0).
\end{align}

Then, by combining with the case of $n_i=1$ (which was discussed in Section \ref{sec:2b} for the payoff $U_{i,1}(r_i(\boldsymbol{x_i},\boldsymbol{x_{-i}}))$), we can see that user $i$'s payoff $U_i(n_i,\boldsymbol{n}_{-i})$ depends on the other users' decisions $\boldsymbol{n_{-i}}=(n_j, \forall j\in\mathcal{I},j\neq i)$. 
  As shown in Fig. 2, each user makes operator selection decision to maximize his utility in one subscription horizon (i.e., the summation of utilities of $T$ data trading horizons). 
	By assuming that a user's demand statistics (as discussed in Section \ref{sec:2a1}) remain the same for each month, the utility function for user $i$ at Stage II over the $T$ trading horizons can be written as	
\begin{equation}\label{eq:u}
\begin{split}
U_i(n_i,\boldsymbol{n}_{-i})=
\!\left\{
    \begin{aligned}
    &T\cdot U_{i,0},\quad\quad\quad\quad\quad\quad\quad\text{if}~~ n_i=0, \\
    &T\cdot U_{i,1}(r_i(\boldsymbol{x_i},\boldsymbol{x_{-i}})),\quad\,\text{if}~~ n_i=1.
    \end{aligned}
        \right.
\end{split}
\end{equation}

Hence, we model the interactions among DTM users as the following non-cooperative game:
\begin{defn}
  An operator selection game is a tuple $\Gamma = (\mathcal{I}, \mathcal{N}, \boldsymbol{U})$ defined by
\begin{itemize}
    \item \emph{Players}: The set $\mathcal{I}$ of users, where user $i \in \mathcal{I}$ is associated with a type $(o_i,p_i,d_{i,h},d_{i,l})$.
    \item \emph{Strategies}: Each player chooses an action (pure strategy) $n_i$ as his choice of operator. 
    The strategy profile of all the players is $\boldsymbol{n} = (n_i,\forall i \in\mathcal{I})$, where $n_i\in\{0,1\}$.
    \item \emph{Payoffs}: The vector $\boldsymbol{U} = (U_{i}, \forall i \in \mathcal{I})$ contains all users' payoffs as defined in (\ref{eq:u}).
\end{itemize}
\end{defn}

\subsection{\rev{Analysis: Nash Equilibrium}}\label{sec:3ec}
We first investigate a user's best response in Section \ref{sec:4BR}, and then characterize the NE in Section \ref{sec:4NE}.

\subsubsection{Best Response Analysis}\label{sec:4BR}
 We assume that all the operators are the same, except that the DTM operator has an additional DTM, so a DTM user will not switch to other operators. In other words, $n_i=1$ is a dominant strategy for the players with $o_i=1$.
For user $i$ with $o_i=0$, he will switch to a DTM operator if his benefit from trade is larger than the switching cost. Hence, we characterize user $i$'s best response in the following proposition.
\begin{pps}\label{thm:br2}
The best response $n_i^{BR}(\boldsymbol{n}_{-i})$ of user $i$ with type $(o_i, p_i, d_{i,h}, d_{i,l})$ is given in (\ref{eq:br2}).
\begin{equation}\label{eq:br2}
\begin{split}
n_i^{BR}(\boldsymbol{n}_{-i})\!=
\!\left\{
    \begin{aligned}
    &0,~\text{if}~~ \frac{(\hat{\pi}(\boldsymbol{n},\theta)\!\!-\!\!\theta)(Q\!\!-\!\!D_l)\!-\!\frac{e}{2}(D_h\!\!+\!\!D_l)}{\kappa(Q-D_l)}\!<\!p_i\!<\!\\
    &~~~~~~~~\frac{\hat{\pi}(\boldsymbol{n},\theta)(D_h\!\!-\!\!Q)\!+\!\frac{e}{2}(D_h\!\!+\!\!D_l)}{\kappa(D_h-Q)}\text{ and }o_i\!=\!0, \\
    &1,~\text{otherwise},
    \end{aligned}
        \right.
\end{split}
\end{equation}
where \revj{$Q$}, $D_l$ and $D_h$ are the mean values of the distributions of \revj{$Q_i$}, $d_{i,l}$ and $d_{i,h}$, respectively.
\end{pps}

The proof of Proposition \ref{thm:br2} is given in Appendix \ref{app:e}.

\subsubsection{Unique NE Characterization}\label{sec:4NE}
Now we study the NE of the above operator selection game. According to the analysis of the equilibrium price $\hat{\pi}(\boldsymbol{n},\theta)$ in Theorem 1 and the best response in Proposition 2, we can obtain the unique NE in Lemma \ref{lem:stage1}. For notation convenience, we first define 
\begin{align}
P_L'(\pi) \triangleq \frac{(\pi\!-\!\theta)(Q\!-\!D_l)\!-\!\frac{1}{2}e(D_h\!+\!D_l)}{\kappa(Q\!-\!D_l)}, \label{eq:lpie}
\end{align}
and
\begin{align}
P_H'(\pi) \triangleq \frac{\pi(D_h\!-\!Q)\!+\!\frac{1}{2}e(D_h\!+\!D_l)}{\kappa(D_h\!-\!Q)},\label{eq:hpie}
\end{align}
\noindent which are the boundary indicators of user types. With these indicators, we can divide the users into different groups according to their user types.

\begin{lem}\label{lem:stage1}
At any NE $\boldsymbol{n}^*(\theta)$ for fixed $\theta$, there exists a unique transaction price $\hat{\pi}(\boldsymbol{n}^*(\theta),\theta)$ such that the following condition is satisfied.
\begin{align}\label{eq:lemma2}
\begin{split}
n_i^*(\theta) =
\!\left\{
    \begin{aligned}
    &0,\quad\text{if}~~ P_L'(\hat{\pi}(\boldsymbol{n}^*(\theta),\theta))\!<\!p_i\!<\!\\
    &~~~~~~~~~~~~P_H'(\hat{\pi}(\boldsymbol{n}^*(\theta),\theta))\text{ and }o_i\!=\!0, \\
    &1,\quad\text{otherwise},
    \end{aligned}
        \right.
\end{split}
\end{align}
where the equilibrium price $\hat{\pi}(\boldsymbol{n}^*(\theta),\theta)$ is the unique solution of the following equation:
\begin{align}\label{eq:price}
\sum_{i\in\{j:p_i\leq P_L(\hat{\pi}(\boldsymbol{n}^*(\theta),\theta))\}}
\!\!\!\!\!\!\!\!\!(Q\!-\!d_{i,l})+\sum_{i\in\{j:p_i\leq P_L'(\hat{\pi}(\boldsymbol{n}^*(\theta),\theta))\}}\!\!\!\!\!\!\!\!\!(Q\!-\!d_{i,l})=\notag\\
\sum_{i\in\{j:p_i\geq P_H(\hat{\pi}(\boldsymbol{n}^*(\theta),\theta))\}}\!\!\!\!\!\!\!\!\!(d_{i,h}\!-\!Q)+\sum_{i\in\{j:p_i\geq P_H'(\hat{\pi}(\boldsymbol{n}^*(\theta),\theta))\}}\!\!\!\!\!\!\!\!\!(d_{i,h}\!-\!Q).
\end{align}
Furthermore, if we assume that both the DTM users' $p_i$ for $i\in \mathcal{I}_1$ and the non-DTM users' $p_i$ for $i \in \mathcal{I} \backslash \mathcal{I}_1$ follow independent uniform distributions in the interval $[0,1]$, we can obtain a closed-form solution for equation (\ref{eq:price}) as
\begin{align}\label{eq:pi}
\hat{\pi}(\boldsymbol{n}^*(\theta),\theta)=\frac{(D_h-Q)\kappa+(Q-D_l)\theta}{D_h-D_l}.
\end{align}
\end{lem}

The proof of Lemma \ref{lem:stage1} is given in Appendix \ref{app:f}.

In the long-run, since user $i$'s operator selection decisions will in turn change his user type $o_i$, the NE of Stages II and III will further evolve. In this paper, we assume that the operator makes Stage I decision based on a short-run user behavior, and characterize the conditions under which the operator will propose a DTM. 
In the next section, to better illustrate the insights in the Stage I problem, we will continue with our assumptions in Lemma \ref{lem:stage1}. For notational simplicity, we will write the equilibrium price $\hat{\pi}(\boldsymbol{n}^*(\theta), \theta)$ as $\hat{\pi}(\theta)$ in the later sections.

\section{Stage I: Operation Fee Optimization}\label{sec:4}
\rev{In this section, we first discuss the DTM operator's profit in Section \ref{sec:4a}. We then formulate its maximization problem in deciding the initial optimal operation fee given its initial market share in Section \ref{sec:4a2}.} Next, we discuss DTM's benefit from proposing this DTM in Section \ref{sec:4b}. 

\subsection{\rev{Model: Operator's Profit}}\label{sec:4a}
The operator's profit consists of three parts: (a) \emph{Base profit}: the difference between the subscription revenue and the cost of providing service from the users, (b) the \emph{total operation fee} charged on the \emph{sellers} for the transacted data quota, and (c) the \emph{overage usage fee} on the users for their data usage that exceeds their quota. 

These three parts are functions of the unit operation fee $\theta$, which should be less than the usage-based unit price $\kappa$  (i.e., $0 \leq \theta \leq \kappa$), or no transaction will happen in the market.
In addition, the operator has a cost of building and maintaining the DTM, which is denoted as $C_b$. Based on Theorem \ref{co:1}, the fractions of sellers, buyers, and the users who do not trade are shown in Fig. \ref{fig:twoline}.

\begin{figure}[t]
\centering
\includegraphics[width=9cm, trim = 2cm 9.5cm 4cm 4cm, clip = true]{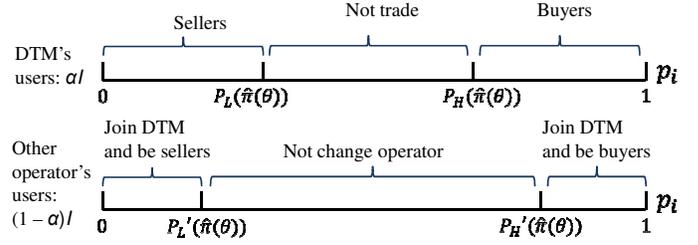}
\vspace{-3mm}
	  \caption{The operation and role selections by the users of both DTM and other operators.}\label{fig:twoline}
\end{figure}

%
%
%

First, the base profit in (a) can be calculated as the product of (a1) the number of users and (a2) the difference between subscription revenue and the cost of providing service for each user. They are described as follows: 
\begin{itemize}
\item (a1) First, according to Table \ref{table:co1}, the number of users is the original number of DTM subscribers $\alpha I$ plus the number of users switched from other operators $(1-\alpha )I(P_L'(\hat{\pi}(\theta))+1-P_H'(\hat{\pi}(\theta)))=(1-\theta/\kappa-e/2(D_h+D_l)/(\kappa(Q-D_l))-e/2(D_h+D_l)/(\kappa(D_h-Q)))(1-\alpha)I$, where $P_L'(\cdot)$ and $P_H'(\cdot)$ are defined in (\ref{eq:lpie}) and (\ref{eq:hpie}), respectively.
\item (a2) Second, the subscription revenue of each user is defined as $\beta$, and the average service cost of each user equals to the product of unit cost $c$ and the average demand $(D_h+D_l)/2$.
\end{itemize}
Hence, the base profit in (a) is:
\begin{align}
P_a(\theta)=&\left(\beta-\frac{c(D_h+D_l)}{2}\right)\Bigg(\alpha I+(1-\frac{\theta}{\kappa}-\frac{e(D_h+D_l)}{2\kappa(Q-D_l)}\notag\\
&-\frac{e(D_h+D_l)}{2\kappa(D_h-Q)})(1-\alpha)I\Bigg).
\end{align}

Second, the total operation fee is charged on the sellers, which includes the original sellers of DTM and the sellers who are switched from other operators. It can be calculated as the product of (b1) the total number of sellers, (b2) the average trading quantity per seller, and the unit operation fee $\theta$. 
\begin{itemize}
\item (b1) For the original sellers of DTM, according to $a_i^*$ in Table \ref{table:co1}, the total number of sellers with $o_i=1$ equals to $\alpha IP_L(\hat{\pi}(\theta))$, where $P_L(\hat{\pi}(\theta))$ is  the fraction of users who can successfully sell some data\footnote{Notice that some users with type $P_L(\hat{\pi}(\theta)) < p_i < P_H(\hat{\pi}(\theta))$ may choose to be a seller but cannot sell any data.} as defined in (\ref{eq:a}). For the sellers who are switched from other operators, according to Lemma 2, the total number of sellers with $o_i=0$ equals to $(1-\alpha) IP_L'(\hat{\pi}(\theta))$, where $P_L'(\cdot)$ is defined in (\ref{eq:lpie}).
\item  (b2) Then, from $q_i^*$ in Table \ref{table:co1}, every seller $i$ with a type $0\leq p_i\leq P_L(\hat{\pi}(\theta))$ will propose a quantity $Q-d_{i,l}$. Thus, the average quantity per seller equals to $Q-D_l$. 
\end{itemize} 
Hence, the total operation fee in (b) is
\begin{align}
P_b(\theta)&=\frac{\alpha\theta I(\hat{\pi}(\theta)\!-\!\theta)}{\kappa}\left(Q\!-\!D_l\right)\notag\\
&+\frac{(1\!-\!\alpha)I\theta(\hat{\pi}(\theta)\!-\!\theta)(Q\!-\!D_l)\!-\!\frac{e}{2}(D_h\!+\!D_l)}{\kappa(Q\!-\!D_l)}(Q\!-\!D_l).
\end{align}

Third, based on Theorem \ref{co:1}, the overage usage fee in (c) are charged on (c1) the sellers and (c2) the DTM users who do not trade.
\begin{itemize}
\item (c1) First, the overage usage fee on the sellers can be calculated as the product of the total number of sellers (which is described in (b1)), the expected overage usage of sellers, and the unit overage usage fee $\kappa$.
Since each original DTM seller will propose $Q-d_{i,l}$ and the users with $p_i<P_L(\hat{\pi}(\theta))$ are sellers, their average overage usage is $D_h-D_l$, and the average probability of high demand is $P_L(\hat{\pi}(\theta))/2=(\hat{\pi}(\theta)-\theta)/(2\kappa)$.  For the sellers who are switched from other operators, according to Lemma 2, their average overage usage is also $D_h-D_l$, but their average probability of high demand is $P_L'(\hat{\pi}(\theta))/2=((\hat{\pi}(\theta)-\theta)(Q-D_l)-e/2(D_h+D_l))/(2\kappa(Q-D_l))$.
Hence, the overage usage fee on the sellers is
\begin{align}
    &P_{c1}(\theta)=\frac{\alpha I (\hat{\pi}(\theta)-\theta)}{\kappa}\kappa\left(D_h-D_l\right)\frac{(\hat{\pi}(\theta)-\theta)}{2\kappa}\notag\\
&+\!\!\frac{(1\!\!-\!\!\alpha) I\kappa\!\left(D_h\!\!-\!\!D_l\right)}{2}\!\!\left(\!\!\frac{ (\hat{\pi}(\theta)\!\!-\!\!\theta)(Q\!\!-\!\!D_l)\!-\!\frac{e}{2}(D_h\!\!+\!\!D_l)}{\kappa(Q-D_l)}\!\right)^2\!\!\!.
\end{align}
\item (c2) Second, the overage usage fee on the DTM users who do not trade can be calculated as the product of the total number of DTM users who do not trade, the expected overage usage of DTM users who do not trade, and the unit overage usage fee $\kappa$. According to the second row in Table 1, the total number of users who do not trade is $\alpha I\theta/\kappa$. For the users who do not trade, their average overage usage is $D_h-Q$. According to the second row in Table 1, the average probability of high demand is $(P_L(\hat{\pi}(\theta))+P_H(\hat{\pi}(\theta)))/2$. Hence, the expected overage usage of users who do not trade is $(\hat{\pi}(\theta)-\theta/2)(D_h-Q)$, and the overage usage fee on the DTM users who do not trade is
\begin{align}
    P_{c2}(\theta)=\theta \alpha I\kappa(D_h-Q)\frac{\hat{\pi}(\theta)-\frac{\theta}{2}}{\kappa}.
\end{align}
\end{itemize}

Overall, the overage usage fee in (c) is
\begin{align}
&P_c(\theta)=P_{c1}(\theta) +P_{c2}(\theta)\notag\\
&=\frac{\alpha I (\hat{\pi}(\theta)-\theta)}{\kappa}\kappa\left(D_h-D_l\right)\frac{(\hat{\pi}(\theta)-\theta)}{2\kappa}\notag\\
&+\!\frac{(1\!-\!\alpha) I\kappa\left(D_h\!-\!D_l\right)}{2}\!\left(\!\frac{ (\hat{\pi}(\theta)\!-\!\theta)(Q\!-\!D_l)\!-\!\frac{1}{2}e(D_h\!+\!D_l)}{\kappa(Q-D_l)}\!\right)^2\notag\\
&+\theta \alpha I\kappa(D_h-Q)\frac{\hat{\pi}(\theta)-\frac{\theta}{2}}{\kappa}.
\end{align}

\subsection{\rev{Problem Formulation: Profit Maximization}}\label{sec:4a2}

Overall, by combining the above three parts, we can write the operator's profit maximization problem as
\begin{equation}
\max_{0\leq\theta\leq \kappa} ~~~ P(\theta) \triangleq P_a(\theta)+P_b(\theta)+P_c(\theta)-C_b.\label{eq:8}
\end{equation}
We can verify that $P(\theta)$ is a concave function of $\theta$, so problem (\ref{eq:8}) is a \emph{convex} optimization problem. As a result, by characterizing the first order condition, we can obtain the unique optimal operation fee in \emph{closed-form}.\footnote{\rev{Overall, after applying backward induction to analyze the three-stage game, the optimal operation fee $\theta^*$ in Stage I (Theorem \ref{thm:22}), the equilibrium operator selection $n_i^*(\theta)$ in Stage II (Lemma \ref{lem:stage1}), and the equilibrium trading decision $\bs{x}_i^* = (a_i^*, q_i^*, \pi_i^*)$ in Stage III (Theorem \ref{co:1}) together constitute the subgame perfect equilibrium \cite{shoham_ma08}.}}
\begin{thm}\label{thm:22}
The operator's optimal seller's operation fee $\theta^*$ is shown in (\ref{eq:thm22}).
\end{thm}
\begin{figure*}[htp] 
\begin{equation}\label{eq:thm22}
\theta^*=\max\left\{0,\min\left\{\kappa,\frac{\kappa}{2}+\frac{(\beta\!-\!\frac{c(D_h\!+\!D_l)}{2})(\alpha\!-\!1)(Q\!-\!D_l)^2-\frac{1}{2}(D_h\!-\!Q)^2(D_h\!-\!D_l)\kappa+\frac{\alpha}{2}(D_h\!-\!Q)(Q\!-\!D_l)^2\kappa}{(2-\alpha)(D_h-Q)(Q-D_l)^2+2\alpha(D_h-Q)^2(Q-D_l)-(D_h-Q)^2(D_h-D_l)}\right\}\right\}.
\end{equation}
\vspace{-1mm}
\end{figure*}
 
\subsection{\rev{Analysis: Benefit of Data Trading Market}}\label{sec:4b}
Next, we will discuss the conditions under which an operator can benefit from proposing a DTM.

To quantify the benefit, we first derive the operator's profit without a DTM.
The operator's profit without a DTM consists of two parts, namely the base profit and the overage usage fee, which are structurally similar to the part (a) and part (c) in Section \ref{sec:4a}. The base profit is the product of the number of users and the difference between subscription revenue and the cost of providing service for each user, so it is equal to $\alpha I\beta-\alpha I(D_h+D_l)c/2$. The overage usage fee is the product of the number of users and the average expected demand among users, so it is equal to $\alpha I\kappa(D_h-Q)/2$. Notice that we do not have operation fee as in part (b) in Section \ref{sec:4a}, because there is no DTM. The profit without DTM as $P_0$ is thus:
\begin{align}
P_0=\frac{\alpha I\kappa(D_h-Q)}{2}-\frac{\alpha I(D_h+D_l)c}{2}+\alpha I\beta.
\end{align}

To characterize whether the operator will benefit from proposing a DTM, we define the threshold market share $\tilde{\alpha}$ in \eqref{eq:alpha} and obtain the following proposition.
\begin{pps}\label{pps:alpha}
If an operator's initial market share $\alpha < \tilde{\alpha}$, then we have $P(\theta^*) > P_0$.
\end{pps}

\begin{figure*}
\begin{align}\label{eq:alpha}
\tilde{\alpha} =\frac{(D_l-Q)(-2\beta+C_b(D_h+D_l)+2\kappa(Q-D_h))}{-2D_h^2\kappa-2\beta D_l+C_bD_l^2+2\beta Q-C_b D_l Q+2\kappa D_l Q-4\kappa Q^2+D_h(C_bD_l-2\kappa D_l-C_bQ+2\kappa Q)}.
\end{align} 
\hrulefill
\end{figure*}

In other words, Proposition \ref{pps:alpha} establishes that an operator with a \emph{small market share} will benefit from proposing a DTM. 
Interestingly, it is in line with the real-world observation that CMHK, the smallest mobile operator in Hong Kong, is the only operator that deploys the DTM\cite{r:cmhkreport}.\footnote{\revj{Consider the case with two identical DTM operators, who have the same initial number of subscribers, subscription fee $\kappa$, unit cost $c$, and subscription revenue per user $\beta$. We can formulate their competition for subscribers as a Bertrand competition, where they will propose the same operation fee $\hat{\theta}$ such that $P(\hat{\theta}) = P_0$ (i.e., both of them cannot gain any profit from deploying such a DTM) at the equilibrium. This also explains why there is only one operator (i.e., CMHK) that deploys a DTM among all the operators in Hong Kong.}}

To better illustrate the insights, we will show the impact of different parameters on the benefit of proposing the DTM through numerical results in Section \ref{sec:sim}.


\section{Performance Evaluation}\label{sec:sim}
In this section, we first provide simulation results to evaluate how an operator can benefit from proposing a mobile DTM, and how different parameters impact on the benefit in Section \ref{sec:sim1}. Then, we evaluate the impact of the market share on the operation fee in Section \ref{sec:sim2}. We show the equilibrium social welfare under different operation fees in Section \ref{sec:sim3}. Finally, we illustrate the users' benefit from this mobile DTM in Section \ref{sec:sim4}.

\subsection{Mobile Data Trading Market's Benefit to Operator}\label{sec:sim1}
We evaluate the operator's benefit by comparing its profits with and without the DTM.

\begin{figure}[t]
\centering
\includegraphics[width = 0.35\textwidth]{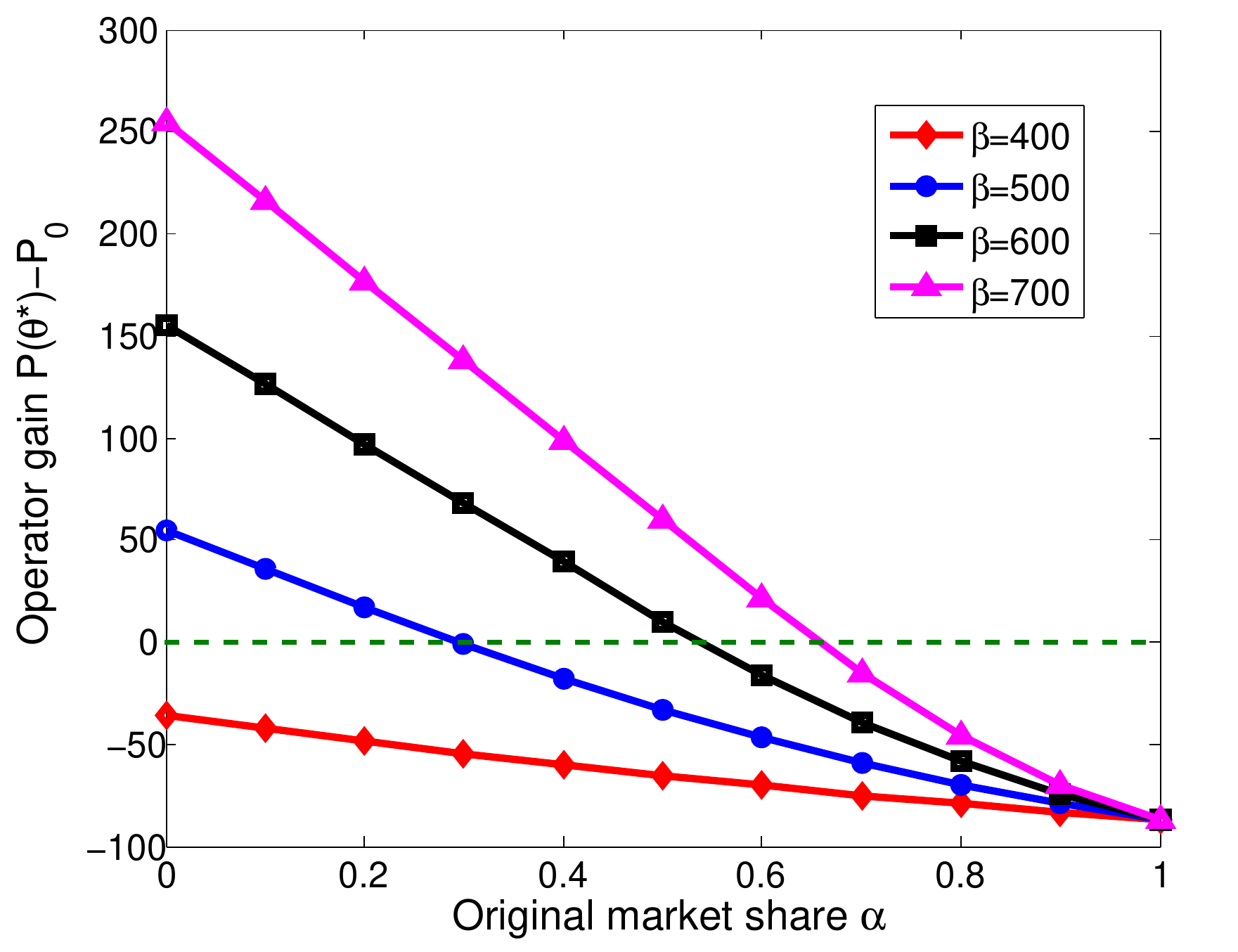}
\vspace{-3mm}
	  \caption{The mobile operator's profit gain of proposing a DTM $P(\theta^*)-P_0$ versus the operator's initial market share $\alpha$ for different subscription fee $\beta$.}\label{fig:sim3}
	  
\end{figure}

In Fig. \ref{fig:sim3}, we plot the profit increment $P(\theta^*)-P_0$ against the operator's initial market share $\alpha$ for different subscription fee $\beta$. In the simulation settings, we choose the data quota $Q=22$, the mean values of the high and low demands $D_h=25$ and $D_l=15$, the unit service cost $c=20$, the switching cost parameter $e=50$, and the DTM maintenance cost $C_b=100$. We observe that the profit increment is increasing in $\beta$. This is because the operator can attract more users by proposing a DTM, hence its profit will increase if the benefit from attracting one more user is larger. We can also see that the profit increment is decreasing in the operator's initial market share $\alpha$, such that there exists a threshold value $\tilde{\alpha}$ (as defined in (\ref{eq:alpha})), below which the operator is willing to propose a DTM. This is because when the initial market share is lower, the operator can attract more users when proposing a DTM. However, if it cannot attract many users, the operator will not benefit because there is maintenance cost for the DTM. For example, when $\beta=400$, proposing the DTM always reduce the profit. When $\beta=600$, the operator will propose a DTM when $\alpha<0.52$. 

\begin{figure}[t]
\centering
\includegraphics[width = 0.35\textwidth]{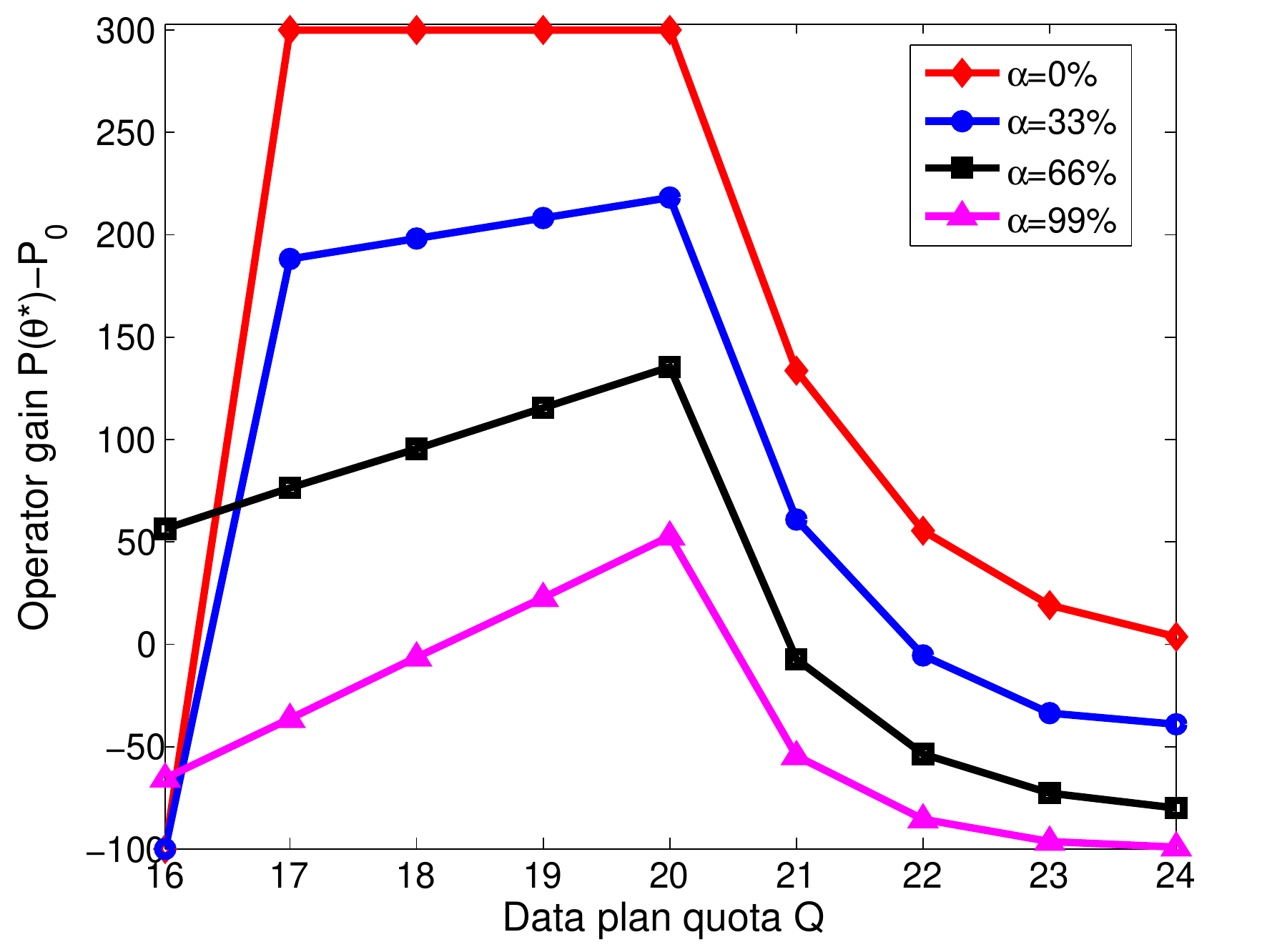}
\vspace{-3mm}
	  \caption{The mobile operator's profit gain of proposing a DTM $P(\theta^*)-P_0$ versus the data quota $Q$  for different $\alpha$.}\label{fig:sim6}
	  
\end{figure}

In Fig. \ref{fig:sim6}, we plot the profit increment $P(\theta^*)-P_0$ against the data quota $Q$ for different initial market share $\alpha$. In the simulation settings, we choose $\beta=500$, $D_h=25$, $D_l=15$, $c=20$, $e=50$, and  $C_b=100$. From Fig. \ref{fig:sim6}, we can see that when the data quota is small ($Q<20$), the operator's profit increment $P(\theta^*)-P_0$ is non-decreasing in $Q$. On the other hand, when the data quota is large ($Q>20$), the operator's profit increment $P(\theta^*)-P_0$ is decreasing in $Q$. It is because increasing data quota will make the sellers more willing to trade but the buyers less willing to trade. When $D_h-Q>Q-D_l$, the sellers' average selling quantity is less than buyers' average buying quantity, hence a larger $Q$ will lead to more trades and increase the operator's revenue. On the other hand, when $D_h-Q<Q-D_l$, the sellers' average selling quantity is more than buyers' average buying quantity, hence a larger $Q$ will lead to less trades and decrease the operator's revenue.

\begin{figure}[t]
    \centering
        \includegraphics[width=0.45\textwidth]{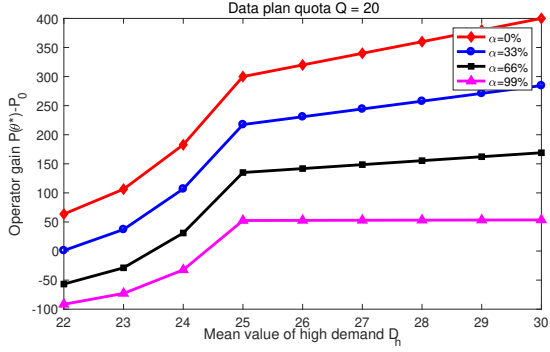} 
    \caption{\rev{The mobile operator's profit gain of proposing a DTM $P(\theta^*)-P_0$ versus the mean value of high demand $D_h$ for different initial market share $\alpha$.}}
    \label{fig:profitvsdh}
\end{figure}	

  \rev{To understand how the user demand affects the DTM operator, we plot the DTM operator’s profit increment $P(\theta^*)-P_0$ against the mean value of high demand $D_h$ in Fig. \ref{fig:profitvsdh}. We can see that when the high demand increases beyond the data plan quota $Q=20$, the operator’s profit increment increases. It is because when the high demand is larger than quota $Q$, there will be more trades, so more users are willing to switch to the DTM operator, which increases its profit.}

\begin{figure}[t]
    \centering
        \includegraphics[width=0.45\textwidth]{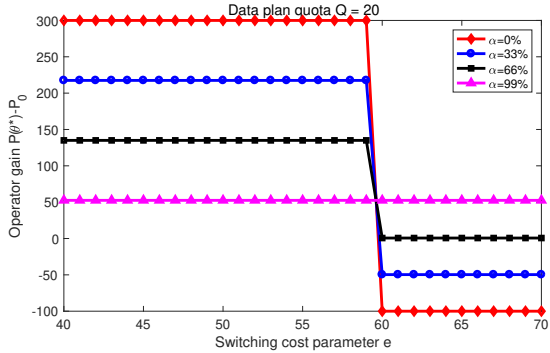} 
    \caption{\rev{The mobile operator's profit gain of proposing a DTM $P(\theta^*)-P_0$ versus the switching cost parameter $e$ for different initial market share $\alpha$.}}
    \label{fig:profitvse}
\end{figure}

\rev{In Fig.~\ref{fig:profitvse}, we plot the DTM operator’s profit increment $P(\theta^*)- P_0$ against the switching cost parameter $e$ for different initial market share $\alpha$. In the simulation settings, we choose usage fee $\kappa = 60$, $\beta = 500$, $D_h = 25$, $D_l = 15$, $c=20$, $Q=20$, and $C_b =100$. From Fig.~\ref{fig:profitvse}, we can see that the operator’s profit increment is piecewise constant with respect to $e$ and there is a sharp drop at $e = \kappa$. It is because when the switching cost parameter $e$ is larger than the overage usage fee $\kappa$, no user will change an operator, and then the DTM’s profit only comes from its original user’s overage usage fee and operation fee, which is not related to the switching cost. Notice that $e$ is the switching cost parameter, and the switching cost of a user equals to the switching cost parameter multiplied by its usage, so that the user’s usage and trading decisions is not related to the switching cost parameter when it is smaller than the overage usage fee $\kappa$.}

\begin{figure}[t]
\centering
\includegraphics[width = 0.35\textwidth]{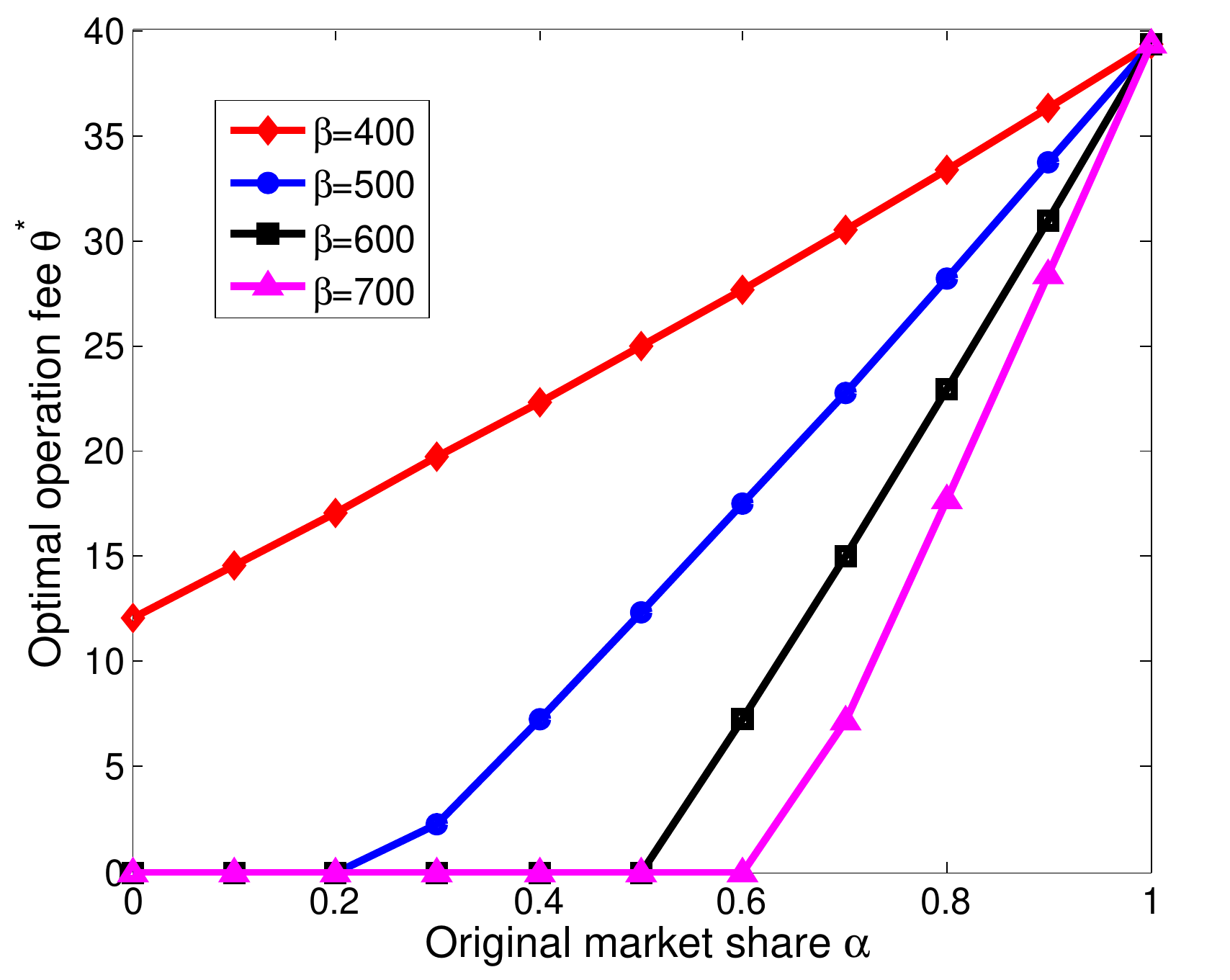}
\vspace{-3mm}
	  \caption{The mobile operator's optimal operation fee $\theta^*$ versus the operator's initial market share $\alpha$ for different subscription fee $\beta$.}\label{fig:sim4}
	 
\end{figure}

\subsection{The Impact of Parameters on Operation Fee}\label{sec:sim2}
In Fig. \ref{fig:sim4}, we plot the optimal operation fee $\theta^*$ against the operator's initial market share $\alpha$ for different subscription fee $\beta$. We observe that $\theta^*$ is increasing in both $\beta$ and $\alpha$. It is because when the benefit of attracting one more user is higher or when proposing the DTM can attract more users, the operator is more willing to decrease its operation fee to better promote the DTM.

\subsection{The Impact of Parameters on Social Welfare} \label{sec:sim3}
We define the total equilibrium user payoff $W_u$ as the sum of all DTM users' payoffs, where $W_u\!\triangleq\!\sum_{i=1}^I U_{i,1}\bigl(r_i(\boldsymbol{x}^*)\bigr)$. 
We define the equilibrium social welfare $W_t$  as the sum of all DTM users' payoffs plus the maximal revenue of the DTM operator (i.e., CMHK), where
$W_t \triangleq \sum_{i=1}^IU_{i,1}\bigl(r_i(\boldsymbol{x}^*)\bigr) + P(\theta^*)$.
 
 \begin{figure}[t]
\centering
\includegraphics[width = 0.35\textwidth]{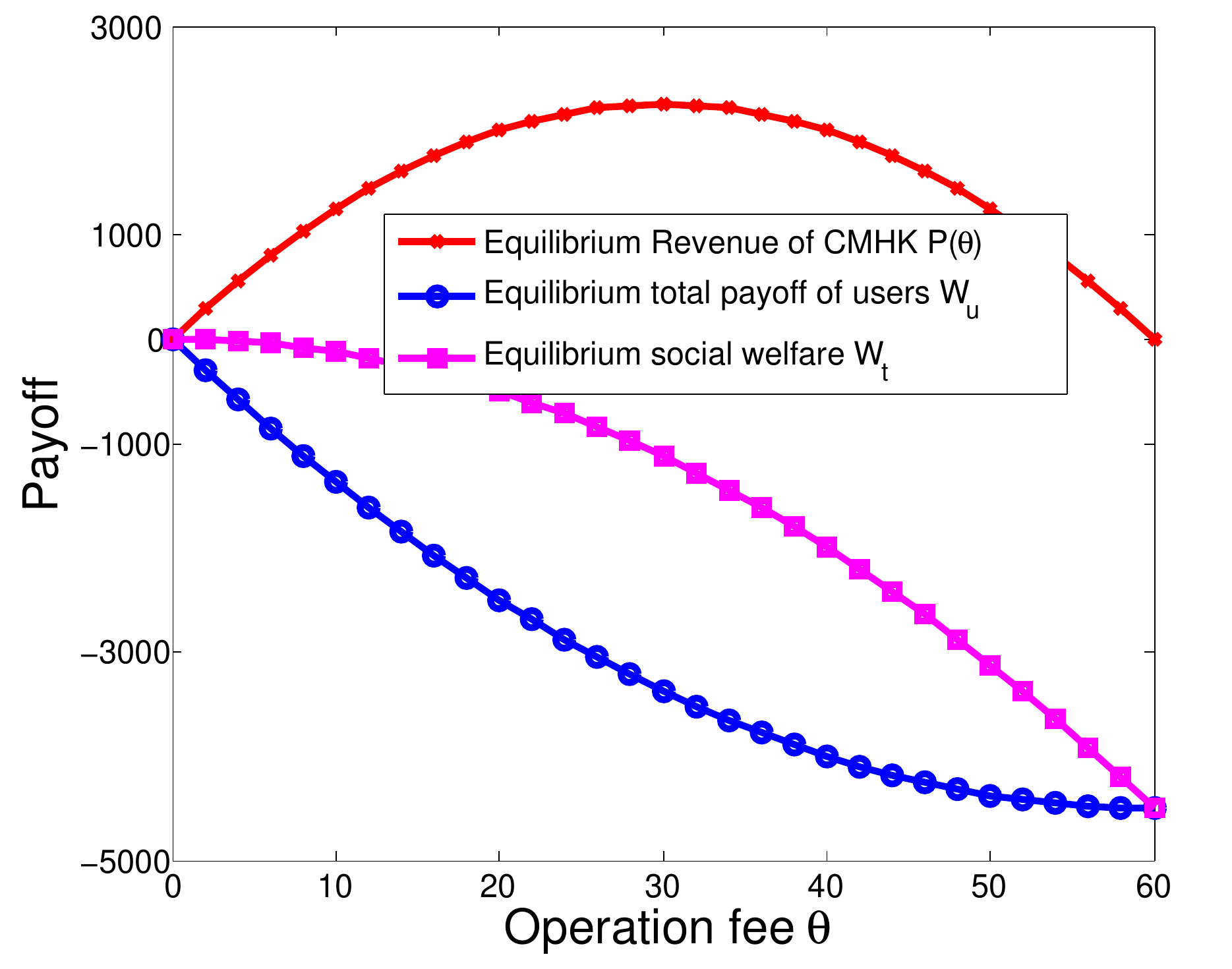}
\vspace{-3mm}
	  \caption{The mobile operator's equilibrium revenue $P(\theta)$, the total payoff of users $W_u$, and the equilibrium social welfare $W_t$ versus the operation fee $\theta$.}\label{fig:sim2}
\end{figure}
 
 In Fig. \ref{fig:sim2}, we show that both the total equilibrium user payoff $W_u$ and the equilibrium social welfare $W_t$ decrease with the operation fee $\theta$. 
 This is because when $\theta$ is larger, users need to pay more to the DTM operator, which results in fewer transactions.
Hence, an operation fee that is too large hurts both the users' payoffs and the DTM operator's revenue.

\subsection{Mobile Data Trading Market's Benefit to Users} \label{sec:sim4}

We evaluate the users' benefits by calculating the gap $G_i$ between DTM user $i$'s payoffs with and without DTM.

In Fig. \ref{fig:sim1}, we plot the payoff increment $G_i$ against the probability of high demand $p_i$, high demand $d_{i,h}$, and low demand $d_{i,l}$. We show that the users with small $p_i$ and large $p_i$ will benefit from trading in the market, while those with medium $p_i$ will not benefit. This is because the users with medium $p_i$ are the most uncertain about whether their usage will exceed quota or not, and hence will not trade data in the market. On the other hand, the users who have less uncertainty about their usage (i.e., the users whose $p_i$ are close to zero or one) will benefit a lot through trading. 

\begin{figure}[t]
\centering
\includegraphics[width = 0.35\textwidth]{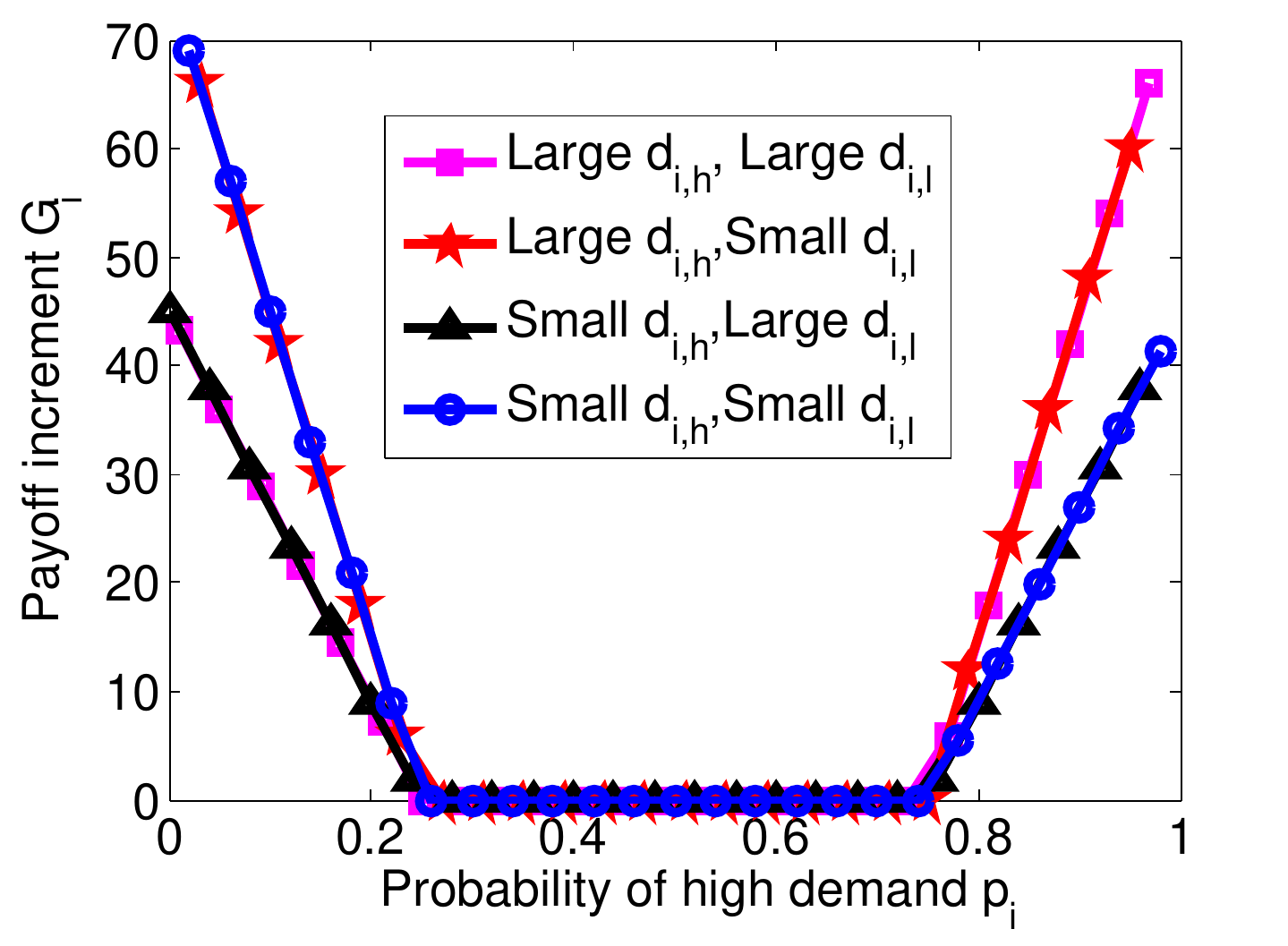}
\vspace{-3mm}
	  \caption{The gap $G_i$ between user $i$'s utilities with and without DTM.}\label{fig:sim1}
	  \vspace{-3mm}
\end{figure}

In addition, the high $p_i$ (i.e., $p_i>0.75$) users will benefit more if they have a larger $d_{i,h}$, and the low $p_i$ (i.e., $p_i<0.25$) users will benefit more if they have a smaller $d_{i,l}$. This is because the users will benefit more when they trade a larger quantity. Based on Theorem 1, a high $p_i$ user will be a buyer and propose a demand of $d_{i,h}-Q$, while a low $p_i$ user will be a seller and propose a supply of $Q-d_{i,l}$.

\section{Conclusion}\label{sec:6}
In this paper, we studied the users' choices of operators and their trading behavior in a data trading market (DTM). 
Our analysis revealed the following interesting insights. First, all the users who want to trade should propose the same price such that the total demand matches the total supply. Second, the non-DTM users who are certain about their usages can benefit more from data trading and will switch to the DTM operator. Third, it is beneficial for a small operator with a low initial market share to propose a DTM to attract subscribers, which is in line with the situation in Hong Kong.

In the future, we would like to understand the countermeasure of other operators in the market competition.
Apart from theoretical analysis, we would also conduct a market survey to understand the users' realistic responses to market dynamics.

\begin{IEEEbiography}
[{\includegraphics[width=1in,clip,keepaspectratio]{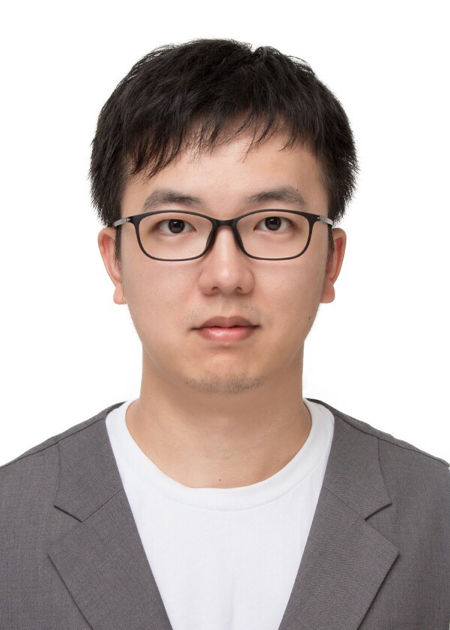}}]
{Junlin Yu} received his Ph.D. degree in the Department of Information Engineering at the Chinese University of Hong Kong in 2017. His research interests include behavioral economical studies and AI in recommendation and marketing, optimization and pricing in financial resource allocation, and optimization in communication and social networks. He is now an algorithm expert in Ant Group.
\end{IEEEbiography}

\begin{IEEEbiography}
[{\includegraphics[width=1in,height=1.25in,keepaspectratio]{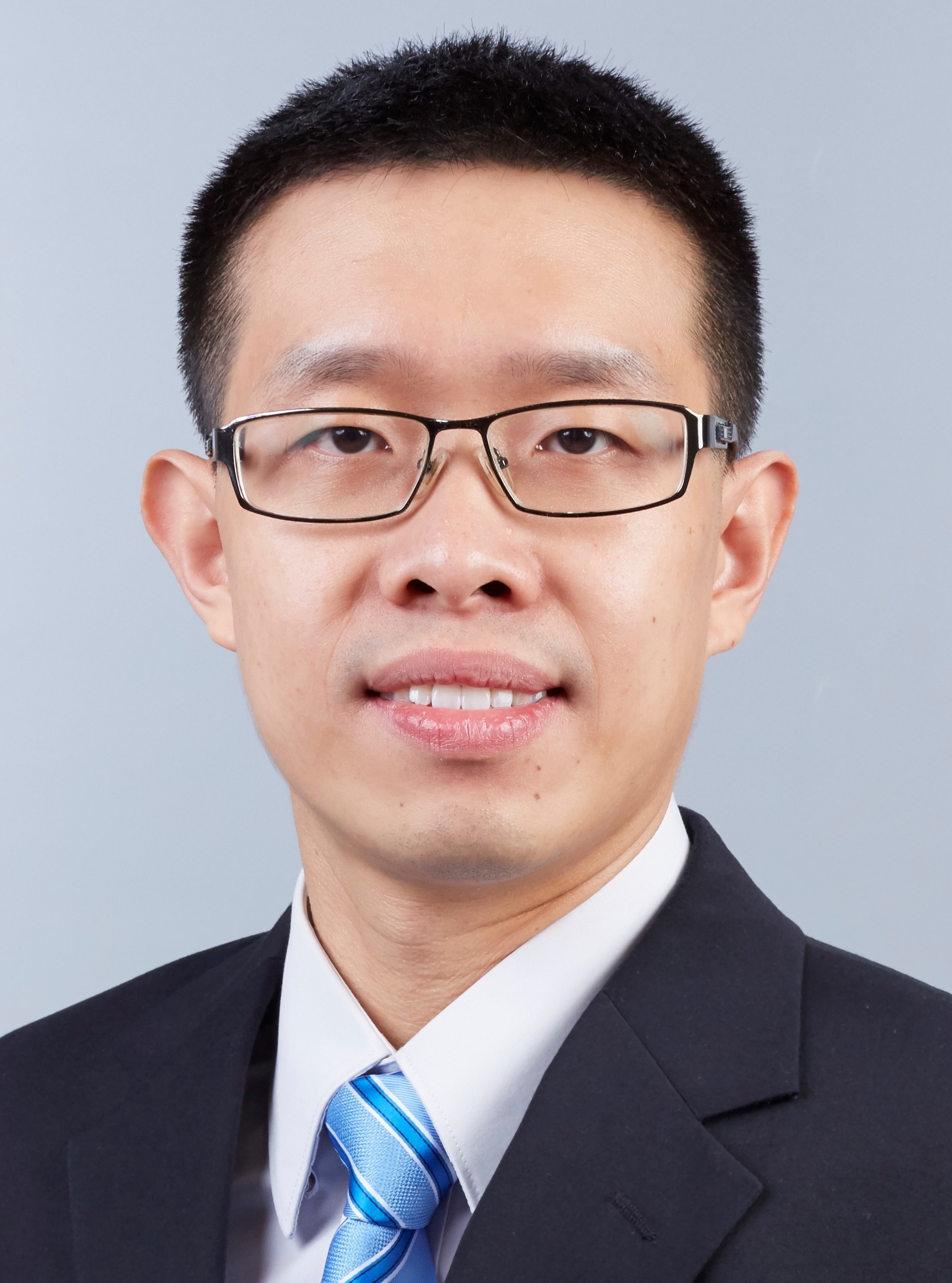}}] 
{Man Hon Cheung} is an Assistant Professor in the Department of Computer Science at the City University of Hong Kong. Previously, he was a Research Assistant Professor at the Department of Information Engineering at the Chinese University of Hong Kong (CUHK). He received the B.Eng. and M.Phil. degrees in Information Engineering from CUHK in 2005 and 2007, respectively, and the Ph.D. degree in Electrical and Computer Engineering from the University of British Columbia (UBC) in 2012. He was awarded the Graduate Student International Research Mobility Award by UBC, and the Global Scholarship Programme for Research Excellence by CUHK. He serves as a Technical Program Committee member in {\it IEEE INFOCOM}, {\it WiOpt}, {\it ICC}, {\it Globecom}, and {\it WCNC}. He is an associate editor of IEEE Communications Letters. %
\end{IEEEbiography}

\begin{IEEEbiography}[{\includegraphics[width=1in,keepaspectratio]{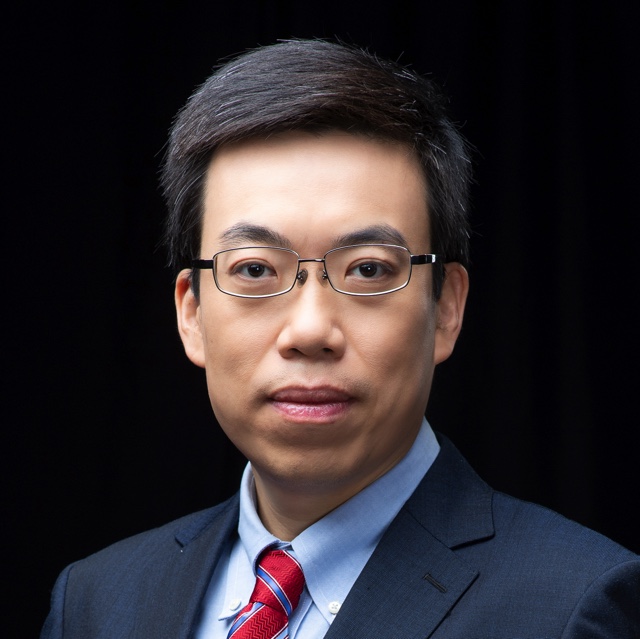}}]
{Jianwei Huang} (F'16) is a Presidential Chair Professor and the Associate Dean of the School of Science and Engineering, The Chinese University of Hong Kong, Shenzhen. He is also the Vice President of Shenzhen Institute of Artificial Intelligence and Robotics for Society. He received the Ph.D. degree from Northwestern University in 2005, and worked as a Postdoc Research Associate at Princeton University during 2005-2007. He has been an IEEE Fellow, a Distinguished Lecturer of IEEE Communications Society, and a Clarivate Analytics Highly Cited Researcher in Computer Science. He has published more than 280 papers in leading international journals and conferences in the area of network optimization and economics, with a total Google Scholar citations of 12,700+ and an H-index of 57. He is the co-author of 9 Best Paper Awards, including IEEE Marconi Prize Paper Award in Wireless Communications in 2011. He has co-authored seven books, including the textbook on "Wireless Network Pricing." He received the CUHK Young Researcher Award in 2014 and IEEE ComSoc Asia-Pacific Outstanding Young Researcher Award in 2009. He has served as an Associate Editor of IEEE Transactions on Mobile Computing, IEEE/ACM Transactions on Networking, IEEE Transactions on Network Science and Engineering, IEEE Transactions on Wireless Communications, IEEE Journal on Selected Areas in Communications - Cognitive Radio Series, and IEEE Transactions on Cognitive Communications and Networking. He has served as the Chair of IEEE ComSoc Cognitive Network Technical Committee and Multimedia Communications Technical Committee. He is the Associate Editor-in-Chief of IEEE Open Journal of the Communications Society. He is the recipient of IEEE ComSoc Multimedia Communications Technical Committee Distinguished Service Award in 2015 and IEEE GLOBECOM Outstanding Service Award in 2010. More detailed information can be found at https://sse.cuhk.edu.cn/en/faculty/huangjianwei.
\end{IEEEbiography}

\clearpage

\begin{center}
\Large{\textbf{Economics of Mobile Data Trading Market}} \\ 
\normalsize{Junlin Yu, Man Hon Cheung, and Jianwei Huang}
\end{center}
\normalsize

\appendices

\section{Details of the market mechanism}\label{app:a}
The detailed market mechanism is shown in Algorithm \ref{algo:dtm}. 

To help explain the transaction rules, we first define the following sets for every user before introducing Algorithm 1.
\begin{equation}
\mathcal{LS}_i \!=\!\! \left\{\!
    \begin{aligned}
    & \{(a_j,\pi_{j},q_{j}): a_j=s ~ \text{and}~ \pi_{j}\!<\!\pi_{i},\\
    &~~~~~~~~~~~~~~~~~~~~~\forall j\!\neq\!i,j\!\in\!\mathcal{I}_1\},\text{ if } a_i\!=\!s,\\
    & \{(a_j,\pi_{j},q_{j}): a_j=s ~\text{and}~\pi_{j}\!\leq\!\pi_{i},\\
    &~~~~~~~~~~~~~~~~~~~~~\forall j\!\neq\!i,j\!\in\!\mathcal{I}_1\}, \text{ if }a_i\!=\!b.\\
    \end{aligned}
    \right.
\end{equation}
If user $i$ is a seller, then set $\mathcal{LS}_i$ refers to the set of sellers, who have higher priorities than user $i$. If user $i$ is a buyer, then set $\mathcal{LS}_i$ refers to the set of sellers who are feasible to be matched with user $i$.
\begin{equation}
\mathcal{HB}_i \!\!=\!\! \left\{\!
    \begin{aligned}
    & \{(a_j,\pi_{j},q_{j}): a_j=b ~\text{and}~\pi_{j}\!\geq\!\pi_{i},\\
    &~~~~~~~~~~~~~~~~~~~~~\forall j\!\neq\! i,j\!\in\!\mathcal{I}_1\},  \text{ if } a_i\!=\!s,\\
    & \{(a_j,\pi_{j},q_{j}): a_j=b ~\text{and}~\pi_{j}\!>\!\pi_{i},\\
    &~~~~~~~~~~~~~~~~~~~~~\forall j\!\neq \!i, j\!\in\!\mathcal{I}_1\},  \text{ if } a_i\!=\!b.\\
    \end{aligned}
    \right.
\end{equation}
If user $i$ is a buyer, then set $\mathcal{HB}_i$ refers to the set of buyers who have higher priorities than user $i$. If user $i$ is a seller, then set $\mathcal{LS}_i$ refers to the set of buyers who are feasible to be matched with user $i$.
\begin{equation}
\mathcal{E}_i \!=\!\{(a_j,\pi_{j},q_{j}\!)\!\!:\!  a_j\!\!=\!a_i ~\text{and}~\pi_{j}\!=\!\pi_{i},\forall j\!\neq \!i,j\!\in\!\mathcal{I}_1\}.
\end{equation}
Set $\mathcal{E}_i$ refers to the set of users who have the same role and the same priority as user $i$.
\begin{align}
\mathcal{ES}_i \!=\!\{(a_j,\pi_{j},q_{j}\!):~& a_j\!\!=\!a_i ~,~\pi_{j}\!=\!\pi_{i},\notag\\
    &\text{and}~q_j<q_i,\forall j\!\neq \!i,j\!\in\!\mathcal{I}_1\}.
\end{align}
Set $\mathcal{ES}_i$ refers to the set of users who have the same role and the same priority as user $i$, and propose smaller quantity than user $i$.
\begin{align}
\tilde{\mathcal{ES}}_i \!=\!\{(&a_j,\pi_{j},q_{j}\!):~ a_j\!\!=\!a_i ~,~\pi_{j}\!=\!\pi_{i},\notag\\
    &\text{and}~q_j\leq\frac{\!\sum\limits_{\boldsymbol{x_k}\in\mathcal{LS}_i}\!q_{k}\!-\!\sum\limits\limits\limits_{\boldsymbol{x_k}\in\mathcal{HB}_i}\!q_{k}\!-\!\sum\limits\limits_{\boldsymbol{x_k}\in\mathcal{ES}_j}\!q_{k}}{|\mathcal{E}_i|-|\mathcal{ES}_j|},\notag\\
    &\forall j\!\neq \!i,j\!\in\!\mathcal{I}_1\}.
\end{align}
Set $\tilde{\mathcal{ES}}_i$ is a subset of $\mathcal{ES}_i$. It refers to the set of users who have the same role and the same priority as user $i$, and propose a very small quantity. Every bid in $\tilde{\mathcal{ES}}_i$ is with small enough quantity such that all the users in $\mathcal{ES}_i$ who propose smaller quantities than it can be satisfied.

If the accumulated buying quantity proposed by the buyers within the set $\mathcal{HB}_i$ is smaller than the accumulated selling quantity proposed by the sellers within the set $\mathcal{LS}_i$, then a seller $i$'s supply cannot be cleared. On the other hand, if $\sum_{\boldsymbol{x_j}\in\mathcal{HB}_i}q_{j}\geq \sum_{\boldsymbol{x_j}\in\mathcal{LS}_i}q_{j}$, then the seller $i$ will \emph{equally share} the demands with the sellers of the same priority (i.e., those in set $\mathcal{E}_i$) \cite{r:poorm,r:emarket}. However, if the equal share is larger than the seller $i$'s supply (i.e., $(\sum_{\boldsymbol{x_j}\in\mathcal{HB}_i}q_{j}-\sum_{\boldsymbol{x_j}\in\mathcal{LS}_i}q_{j})/|\mathcal{E}_i|>q_{i}$ for some seller $i$), then his $r_i(\boldsymbol{x_i},\boldsymbol{x_{-i}})=q_{i}$, and the ``burden left'' is averaged over the other sellers with the same price. We will continue with this procedure until each seller $j$ that proposes the same price has an $r_j(\boldsymbol{x_j},\boldsymbol{x_{-j}})\leq q_{j}$. The same rule also applies when two buyers propose the same price.

Based on the transaction rule, among the sellers or buyers proposing the same price, the users who propose a smaller quantity will get all their quantity transacted before the users who propose a higher quantity. As an example in Fig.~\ref{fig:1}, consider three buyers, $1$, $2$, and $3$, proposing the same buying price of \$14, where $q_1=3$ GB, $q_2=4$ GB, and $q_3=8$ GB. Hence their total demand is 15 GB as shown on the second row, but there are only 5 GB left in seller's market (with the selling price of \$13) that can be allocated to them. According to the transaction rule in (\ref{eq:realization1}) and (\ref{eq:realization2}), they will equally divide the 5 GB, i.e., $r_1(\boldsymbol{x})=r_2(\boldsymbol{x})=r_3(\boldsymbol{x})=5/3$ GB. Consider a different scenario where the three buyers' demands are $q_1=1$ GB, $q_2=6$ GB, and $q_3=8$ GB, then the allocations are $r_1(\boldsymbol{x})=1$ GB and $r_2(\boldsymbol{x})=r_3(\boldsymbol{x})=2$ GB. This is because with the result under the equal division ($5/3$) exceeds user $1$'s demand, hence the exceeded part is equally shared by the remaining two buyers. \QEDB

\begin{algorithm}[t] \label{algo:dtm}
 \caption{Mobile Data Trading Market Allocation Mechanism}
\For{$i$ \text{in} $\mathcal{I}$}
{ User $i$ submit his bid $\boldsymbol{x_i}=(a_i,\pi_i,q_i)$ to the DTM operator.
}
\For{$i$ in $\mathcal{I}$}
{ 
\If{$a_i=s$}{
User $i$ can get $r_i(\boldsymbol{x_i},\boldsymbol{x}_{-i})$ allocated, where
\begin{align}\label{eq:realization1}
&r_i(\boldsymbol{x_i},\boldsymbol{x}_{-i})=\notag\\&\min\left\{q_i,\frac{\!\sum\limits_{\boldsymbol{x_j}\in\mathcal{HB}_i}\!q_{j}\!-\!\sum\limits_{\boldsymbol{x_j}\in\mathcal{LS}_i}\!q_{j}\!-\!\sum\limits\limits_{\boldsymbol{x_j}\in\tilde{\mathcal{ES}}_i}\!q_{j}}{|\mathcal{E}_i|-|\tilde{\mathcal{ES}}_i|}\right\}.
\end{align}
}
\ElseIf{$a_i=b$}
{
User $i$ can get $r_i(\boldsymbol{x_i},\boldsymbol{x}_{-i})$ allocated, where
\begin{align}\label{eq:realization2}
&r_i(\boldsymbol{x_i},\boldsymbol{x}_{-i})=\notag\\&\min\left\{q_i,\frac{\!\sum\limits_{\boldsymbol{x_j}\in\mathcal{LS}_i}\!q_{j}\!-\!\sum\limits\limits\limits_{\boldsymbol{x_j}\in\mathcal{HB}_i}\!q_{j}\!-\!\sum\limits\limits_{\boldsymbol{x_j}\in\tilde{\mathcal{ES}}_i}\!q_{j}}{|\mathcal{E}_i|-|\tilde{\mathcal{ES}}_i|}\right\}.
\end{align}
}
}
Calculate the operator's income of gap between the selling and buying prices as follows:
\begin{align}
P_{gap}=\sum_{i\in\{j\in\mathcal{I}_1:a_j=s\}}\!\!\!\pi_i r_i(\boldsymbol{x})-\!\!\!\sum_{i\in\{j\in\mathcal{I}_1:a_j=b\}}\!\!\!\pi_i r_i(\boldsymbol{x}).
\end{align}
\end{algorithm}

\section{Proof of Proposition 1}\label{app:b}
Given a strategy profile $\boldsymbol{x}$, by (\ref{eq:stp}) and (\ref{eq:btp}), we can obtain the corresponding transaction selling price $\hat{\pi}_s$ and transaction buying price $\hat{\pi}_b$. We first study the user $i$'s optimal decision of quantity, given his role and price, and other users' strategies. From (\ref{eq:thm2}), we can find the optimal proposed quantity $q_i^{BR}(\boldsymbol{x}_{-i})$ for a user given his price $\pi_i^{BR}(\boldsymbol{x}_{-i})$ and $a_i^{BR}(\boldsymbol{x}_{-i})$ as
\begin{align}
     q_i^{BR}(\boldsymbol{x}_{-i})\left\{\!\!
    \begin{aligned}
    &=d_{i,h}\!-\!Q, ~ \textstyle \text{if} \quad p_i\!\geq\!\frac{\pi_i^{BR}(\boldsymbol{x}_{-i})}{\kappa},\\
    &=Q\!-\!d_{i,l},~~ \textstyle \text{if}\quad p_i\!\leq\!\frac{\pi_i^{BR}(\boldsymbol{x}_{-i})\!-\!\theta}{\kappa},\\
    &\,\rev{0}, ~~~\textstyle \text{if}\quad \frac{\pi_i^{BR}(\boldsymbol{x}_{-i})\!-\!\theta}{\kappa}\!<\!p_i\!<\!\frac{\pi_i^{BR}(\boldsymbol{x}_{-i})}{\kappa}. 
    \end{aligned}
    \right.
\end{align}

Then we study the decision of price and role. According to the utility function in (\ref{eq:thm2}) and the transaction rule in (\ref{eq:realization1}) and (\ref{eq:realization2}), the low $p_i$ users will choose to be a seller and the high $p_i$ users will choose to be a buyer, i.e., 
\begin{align}
     a_i^{BR}(\boldsymbol{x}_{-i})\left\{\!\!
    \begin{aligned}
    &=b,  ~\text{if}\quad p_i\!\geq\!\frac{\pi_i^{BR}(\boldsymbol{x}_{-i})}{\kappa},\\
    &=s, ~\text{if}\quad p_i\!\leq\!\frac{\pi_i^{BR}(\boldsymbol{x}_{-i})\!-\!\theta}{\kappa}.\\
    \end{aligned}
    \right.
\end{align}

Due to the discontinuity in the seller's utility function on the first line of (\ref{eq:thm2}), the utility $U_{i,1}(r_i(\boldsymbol{x}_i,\boldsymbol{x}_{-i}))$ for a seller first increases in $\pi_i$ when $\pi_i\in[0,\hat{\pi}_s-\epsilon)$, then has one or two discontinuous decreases (i.e., discontinuous jump downwards) in $\pi_i$ when $\pi_i\in[\hat{\pi}_s-\epsilon,\hat{\pi}_s+\epsilon]$\footnote{There is a discontinuous decrease in utility between $\pi_i=\hat{\pi}_s$ and  $\pi_i=\hat{\pi}_s+\epsilon$, but we still need to discuss whether there is a discontinuous decrease between $\pi_i=\hat{\pi}_s-\epsilon$ and  $\pi_i=\hat{\pi}_s$.}, and then stays flat in $\pi_i$ when $\pi_i\in(\hat{\pi}_s+\epsilon,\kappa]$. 
Due to the equal division rule, the sellers with smaller selling quantity $Q-d_{i,l}$ (the sellers whose supplies satisfy the condition on Line 1 in (\ref{eq:thm2})) will be satisfied first, which means that their supplies will be cleared by proposing a selling price $\hat{\pi}_s$. Hence, the unique discontinuous decrease is between the point $\pi_i=\hat{\pi}_s$ and the point $\pi_i=\hat{\pi}_s+\epsilon$, i.e.,
$$U_{i,1}\!(r_i(\!(a_i^{B\!R}\!,\hat{\pi}_s-\epsilon,q_i^{B\!R}),\boldsymbol{x}_{-i})\!)\!<\!U_{i,1}\!(r_i(\!(a_i^{B\!R}\!,\hat{\pi}_s,q_i^{B\!R}),\boldsymbol{x}_{-i})\!),$$
$$U_{i,1}\!(r_i(\!(a_i^{B\!R}\!,\hat{\pi}_s,q_i^{B\!R}),\boldsymbol{x}_{-i})\!)\!>\!U_{i,1}\!(r_i(\!(a_i^{B\!R}\!,\hat{\pi}_s+\epsilon,q_i^{B\!R}),\boldsymbol{x}_{-i})\!).$$
In this case, the optimal price is at the transaction selling price $\hat{\pi}_s$. On the other hand, the sellers with a larger selling quantity $Q-d_{i,l}$ (the sellers whose supplies satisfy the condition on Line 2 in (\ref{eq:thm2})) can decrease their price by a small amount $\epsilon$ to beat the sellers who propose $\hat{\pi}_s$ and fully get their quantity transacted. In other words, there is also a discontinuous decrease between the point $\pi_i=\hat{\pi}_s-\epsilon$ and the point $\pi_i=\hat{\pi}_s$, i.e.,
$$U_{i,1}\!(r_i(\!(a_i^{B\!R}\!,\hat{\pi}_s-\epsilon,q_i^{B\!R}),\boldsymbol{x}_{-i})\!)\!>\!U_{i,1}\!(\!r_i(\!(a_i^{B\!R}\!,\hat{\pi}_s,q_i^{B\!R}),\boldsymbol{x}_{-i})\!),$$
$$U_{i,1}\!(r_i(\!(a_i^{B\!R}\!,\hat{\pi}_s,q_i^{B\!R}),\boldsymbol{x}_{-i})\!)\!>\!U_{i,1}\!(r_i(\!(a_i^{B\!R}\!,\hat{\pi}_s+\epsilon,q_i^{B\!R}),\boldsymbol{x}_{-i})\!).$$
In this case, the optimal price is at $\hat{\pi}_s-\epsilon$.
By the above analysis, we obtain the results on the first and second lines of (\ref{eq:thm2}).

Similarly, due to the discontinuity in the buyer's utility function on the second line of (\ref{eq:thm2}), the utility $U_{i,1}(r_i(\boldsymbol{x}_i,\boldsymbol{x}_{-i}))$ for a buyer is first flat in $\pi_i$ when $\pi_i\in[0,\hat{\pi}_b-\epsilon)$, then has one or two discontinuous increases in $\pi_i$ when $\pi_i\in[\hat{\pi}_b-\epsilon,\hat{\pi}_b+\epsilon]$, and then decreases in $\pi_i$ when $\pi_i\in(\hat{\pi}_b+\epsilon,\kappa]$. By a similar analysis, we can show that the buyers with smaller buying quantity $d_{h,l}-Q$ (the buyers whose demands satisfy the condition on Line 3 in (\ref{eq:thm2})) have an optimal price at $\hat{\pi}_b$, and the buyers with larger buying quantity $d_{h,l}-Q$ (the buyers whose demands satisfy the condition on Line 4 in (\ref{eq:thm2})) have an optimal price at $\hat{\pi}_b+\epsilon$. Hence, we obtain the results on the third and fourth lines of (\ref{eq:thm2}).

Finally, for the users with type $(\hat{\pi}_s-\theta)/\kappa<p_i<\hat{\pi}_b/\kappa$, \rev{the utility $U_i((a_i,\pi_i,q_i^{BR}),\boldsymbol{x}_{-i})$ stays flat in $\pi_i$, so we choose $\pi_i^{BR}(\boldsymbol{x}_{-i}) = 0$.}


By the above analysis, we obtain the results in (\ref{eq:thm2}).\QEDB

\section{Proof of Lemma 1}\label{app:c}
In this proof, we do not consider the users who are not willing to trade. In other words, we only consider the user $i\in\tilde{\mathcal{I}}=\{j:j\in\mathcal{I},r_j(\boldsymbol{x})>0\}$.
First, we show that the following lemma holds at the equilibrium. 
\begin{lem}\label{lem:pps21}
For any two sellers or two buyers proposing the same price, both of them can get all their quantity transacted. That is, if $a_j^*=a_k^*$ and $\pi_j^*=\pi_k^*$, we have $r_j(\boldsymbol{x}^*)=q_j^*$ and $r_k(\boldsymbol{x}^*)=q_k^*$.
\end{lem}
\begin{proof}
We prove the lemma by contradiction. Assume that there exist $j$ and $k$ such that $$a_j^*=a_k^*=s, \pi_j^*=\pi_k^*, r_j(\boldsymbol{x}^*)< q_j^*,$$ then based on (\ref{eq:1}), we have
\begin{align}
&U_{j,1}(r_j((s,\pi_j^*,q_j^*),\boldsymbol{x}_{-j}^*))\notag\\
&=(\pi_j^*-\theta)r_j(\boldsymbol{x}^*)+p_i\kappa (\revj{Q_j}-r_j((s,\pi_j^*,q_j^*)-d_{j,h}),
\end{align}
and 
\begin{align}
U_{j,1}(r_j&((s,\pi_j^*\!-\!\epsilon,q_j^*),\boldsymbol{x}_{-j}^*))=\notag\\
 &(\pi_j^*\!-\!\epsilon\!-\!\theta)\min\{q_j^*,r_j(\boldsymbol{x}^*)\!+\!r_k(\boldsymbol{x}^*)\}\notag\\
&+p_i\kappa (\revj{Q_j}-\min\{q_j^*,r_j(\boldsymbol{x}^*)+r_k(\boldsymbol{x}^*)\}-d_{j,h}),
\end{align}
where $\epsilon$ is an extremely small positive number. 

According to (\ref{eq:thm2}), we have $p_i\kappa<\pi_i^*-\theta$ for a seller. Hence, when $\epsilon\rightarrow 0$, we have $$U_{j,1}(r_j((s,\pi_j^*-\epsilon,q_j^*),\boldsymbol{x}_{-j}^*))>U_{j,1}(r_j((s,\pi_j^*,q_j^*),\boldsymbol{x}_{-j}^*)),$$ which contradicts with (\ref{eq:def4}) in Definition 5.
    
The buyer's case is similar to seller's case, hence we omit the proof here.
\end{proof}

Next we prove that Lemma \ref{lem:pps22} also holds at the equilibrium.
\begin{lem}\label{lem:pps22}
Any two users $j$ and $k$ who want to trade will propose the same price, i.e., 
\begin{equation}
\pi_j^*=\pi_k^*, ~\forall j,k\in\tilde{\mathcal{I}}.
\end{equation}
\end{lem}

\begin{proof}
We prove the lemma by contradiction.
We also start with the sellers' case by proving
\begin{equation}
\pi_j^*=\pi_k^*, ~\forall j,k\in\{i:i\in\mathcal{I},a_i^*=s\}.
\end{equation}

Assume that there exist $j$ and $k$ such that $$a_j^*=a_k^*=s, \pi_j^*<\pi_k^*,$$ then we know that $r_j(\boldsymbol{x}^*)=q_j^*$, and have
$$U_{j,1}(r_j((s,\pi_j^*,q_j^*),\boldsymbol{x}_{-j}^*))=(\pi_j^*\!-\!\theta)q_j^*+p_i\kappa (\revj{Q_j}\!-\!q_j^*\!-\!d_{j,h})$$
and 
$$U_{j,1}(r_j((s,\pi_j^*\!+\!\epsilon,q_j^*),\boldsymbol{x}_{-j}^*))\!=\!(\pi_j^*+\epsilon-\theta)q_j^*+p_i\kappa (\revj{Q_j}-q_j^*-d_{j,h}).$$

Hence, we have $$U_{j,1}(r_j((s,\pi_j^*+\epsilon,q_j^*),\boldsymbol{x}_{-j}^*))>U_{j,1}(r_j((s,\pi_j^*,q_j^*),\boldsymbol{x}_{-j}^*)),$$ which contradicts with (\ref{eq:def4}) in Definition 5.

Similarly, we know that
\begin{equation}
\pi_j^*=\pi_k^*, ~\forall j,k\in\{i:i\in\mathcal{I},a_i^*=b\}.
\end{equation}

Next, we prove that $\pi_j^*=\pi_k^*, ~\forall j,k\in\mathcal{I}$, also by contradiction. When $\pi_j^*<\pi_k^*, ~a_j^*=s,$ and $a_k^*=b$, we have
$$U_{j,1}(r_j((s,\pi_j^*+\epsilon,q_j^*),\boldsymbol{x}_{-j}^*))>U_{j,1}(r_j((s,\pi_j^*,q_j^*),\boldsymbol{x}_{-j}^*))$$ and $$U_{j,1}(r_k((b,\pi_k^*+\epsilon,q_k^*),\boldsymbol{x}_{-k}^*))>U_{j,1}(r_k((b,\pi_k^*,q_k^*),\boldsymbol{x}_{-k}^*)),$$ which contradicts with (\ref{eq:def4}) in Definition 5.

\end{proof}

By Lemma \ref{lem:pps21}, we know that all users can get their proposed quantity fully transacted, which will only happen when the total proposed buying quantity equals the total proposed selling quantity. By (\ref{eq:thm2}), we know that the users with type $p_i<(\hat{\pi}(\boldsymbol{n},\theta)-\theta)/\kappa$ will trade as a seller, and the users with type $p_i>\hat{\pi}(\boldsymbol{n},\theta)/\kappa$ will trade as a buyer. By Lemma \ref{lem:pps22}, we know that all the buyers and sellers will propose the same price. 
Combining the above analysis, we obtain the result in (\ref{eq:lem11}) and (\ref{eq:lem12}).\QEDB

\section{Proof of Theorem 1}\label{app:d}
First, by (\ref{eq:thm2}), (\ref{eq:lem11}), and (\ref{eq:lem12}), we can obtain the results in Table \ref{table:co1}.
Then we try to obtain $\hat{\pi}(\boldsymbol{n},\theta)$ in Table \ref{table:co1}. According to (\ref{eq:lem12}), we know that the proportion of sellers is $(\hat{\pi}(\boldsymbol{n},\theta)-\theta)/\kappa$ and the proportion of buyers is $1-\hat{\pi}(\boldsymbol{n},\theta)/\kappa$. Based on (\ref{eq:thm2}), we know that every seller proposes a quantity of $\revj{Q_i}-d_{i,l}$ and every buyer proposes a quantity of $d_{i,h}-\revj{Q_i}$. Hence, the total selling quantity
\begin{equation}
    Q_s=\sum_{i\in\{j: p_j\leq P_L(\hat{\pi}(\boldsymbol{n},\theta))\}}(\revj{Q_i}-d_{i,l}),
\end{equation}
and the total buying quantity 
\begin{equation}
    Q_b=\sum_{i\in\{j: p_j\geq P_H(\hat{\pi}(\boldsymbol{n},\theta))\}}(d_{i,h}-\revj{Q_i}).
\end{equation}
Based on (\ref{eq:lem11}), we know that the total selling quantity equals to the total buying quantity, i.e.,
\begin{equation}
Q_s=Q_b,
\end{equation}
which is equivalent to (\ref{eq:eqp}). \QEDB

\section{Proof of Proposition 2}\label{app:e}
First, for the users with $o_i=1$, according to (\ref{eq:thm2}) and (\ref{eq:epsilon}), we have $U_{i,1}(r_i((s,0,0),\boldsymbol{x}_{-i}))>U_{i,0}$ for all $\boldsymbol{x}_{-i}$. Hence, we have
\begin{equation}
n_i^{BR}(\boldsymbol{n}_{-i})=1
\end{equation}
for all users with $o_i=1$.

Next, for the users with $o_i=0$, according to Theorem 1, a user $i$ with $p_i\in[0,P_L(\hat{\pi}(\boldsymbol{n},\theta))]\cup[P_H(\hat{\pi}(\boldsymbol{n}),\theta),1]$ will trade if he switches to the DTM operator. According to (\ref{eq:thm2}) and (\ref{eq:epsilon}), the corresponding utility $U_{i,1}(r_i(\boldsymbol{x}^*(1,\boldsymbol{n}_{-i})))>U_{i,0}$ only when\footnote{Since the total number of users $I\rightarrow\infty$, we can approximate $\hat{\pi}(\boldsymbol{n},\theta)=\hat{\pi}((\hat{n}_i,\boldsymbol{n}_{-i}),\theta) ~\forall \hat{n}_i\in\mathcal{N}$ for any $\boldsymbol{n}_{-i},i\in\mathcal{I}$.} 
\begin{align}
    p_i\!<\!\frac{(\hat{\pi}(\boldsymbol{n},\theta)\!-\!\theta)(Q\!-\!D_l)\!-\!\frac{e}{2}(D_h\!+\!D_l)}{\kappa(Q\!-\!D_l)}\notag\\
    \text{ or }~p_i\!>\!\frac{\hat{\pi}(\boldsymbol{n},\theta)(D_h\!-\!Q)\!+\!\frac{e}{2}(D_h\!+\!D_l)}{\kappa(D_h\!-\!Q)}.
\end{align}
For the other users with $o_i=0$, the utility $U_{i,0}$ is higher than $U_{i,1}(r_i(\boldsymbol{x}^*(1, \boldsymbol{n}_{-i})))$.
Hence, by combining the two, we can obtain the results in Proposition 2.\QEDB

\section{Proof of Lemma 2}\label{app:f}
First, according to the analysis in Theorem 1, all the sellers and buyers will propose a same market price $\hat{\pi}(\boldsymbol{n}^*(\theta),\theta)$ satisfying (\ref{eq:price}). By substituting (\ref{eq:lpie}) and (\ref{eq:hpie}) into (\ref{eq:br2}), we can obtain (\ref{eq:lemma2}). 

Then, since the LHS of (\ref{eq:price}) is increasing in $\hat{\pi}(\boldsymbol{n}^*(\theta),\theta)$ and the RHS of (\ref{eq:price}) is decreasing in $\hat{\pi}(\boldsymbol{n}^*(\theta),\theta)$, we can see that there exists a unique solution for (\ref{eq:price}).

Next, if both the DTM users' $p_i$ for $i\in \mathcal{I}_1$ and the non-DTM users' $p_i$ for $i \in \mathcal{I} \backslash \mathcal{I}_1$ follow independent uniform distributions in the interval $[0,1]$, then equation (\ref{eq:price}) can be rewritten as 
\begin{align}\label{eq:last}
    &\alpha \frac{\hat{\pi}(\boldsymbol{n}^*(\theta),\theta)-\theta}{\kappa}(Q-D_l)\notag\\
    &+(1\!-\!\alpha)\frac{(\hat{\pi}(\boldsymbol{n}^*(\theta),\theta)\!-\!\theta)(Q\!-\!D_l)\!-\!\frac{1}{2}e(D_h\!+\!D_l)}{\kappa(Q-D_l)}(Q\!-\!D_l)\notag\\
    &=\alpha(1-\frac{\hat{\pi}(\boldsymbol{n}^*(\theta),\theta)}{\kappa})(D_h-Q)\notag\\
    &+(1\!-\!\alpha)(1\!-\!\frac{\hat{\pi}(\boldsymbol{n}^*(\theta),\theta)(D_h\!-\!Q)\!+\!\frac{1}{2}e(D_h\!+\!D_l)}{\kappa(D_h-Q)})(D_h\!-\!Q).
\end{align}
By solving (\ref{eq:last}), we obtain the results in (\ref{eq:pi}).\QEDB

\end{document}